\documentclass[11pt,reqno]{amsart}
\usepackage{amssymb,mathrsfs,graphicx,enumerate}
\usepackage{amsmath,amsfonts,amssymb,amscd,amsthm,bbm}
\usepackage[retainorgcmds]{IEEEtrantools}
\usepackage{colortbl}
\allowdisplaybreaks
\title[Quadratically separable states to the Lohe tensor model]{Existence and emergent dynamics of quadratically separable states to the Lohe tensor model}

\author[Ha]{Seung-Yeal Ha}
\address[Seung-Yeal Ha]{\newline Department of Mathematical Sciences and Research Institute of Mathematics, \newline Seoul National University, Seoul 08826 and \newline
Korea Institute for Advanced Study, Hoegiro 85, Seoul, 02455, Republic of Korea} 
\email{syha@snu.ac.kr}
 
\author[Kim]{Dohyun Kim}
\address[Dohyun Kim]{\newline School of Mathematics, Statistics and Data Science,  \newline Sungshin Women's University, Seoul 02844, Republic of Korea}
\email{dohyunkim@sungshin.ac.kr}

\author[Park]{Hansol Park}
\address[Hansol Park]{\newline Department of Mathematical Sciences, \newline Seoul National University, Seoul 08826, Republic of Korea}
\email{hansol960612@snu.ac.kr}

\newtheorem{theorem}{Theorem}[section]
\newtheorem{lemma}{Lemma}[section]
\newtheorem{corollary}{Corollary}[section]
\newtheorem{proposition}{Proposition}[section]
\newtheorem{remark}{Remark}[section]

\newtheorem{definition}{Definition}[section]

\newcommand{\bbr}{\mathbb R}

\newcommand{\bbs}{\mathbb S}
\newcommand{\bbz}{\mathbb Z}

\newcommand{\bbc}{\mathbb C}

\newcommand{\dg}{  {\dagger}}
\renewcommand{\d}{ {\textup{d}}}
\newcommand{\dt}{   {\textup{d}t}}
\newcommand{\tF}{  {\textup{F}}}
\newcommand{\kp}{ {\kappa}}
\newcommand{\Un} {\mathbf{U}(n)}
\newcommand{\Um} {\mathbf{U}(m)}
\newcommand{\SOn} {\mathbf{SO}(n)}
\newcommand{\SOm} {\mathbf{SO}(m)}

\newcommand{\tr}{\textup{tr}}
\renewcommand{\Re}{\textup{Re}}
\newcommand{\mi}{\mathrm{i}}

\begin{document}

\date{\today}



\thanks{\textbf{Acknowledgment.} The work of S.-Y. Ha is supported by National Research Foundation of Korea (NRF-2020R1A2C3A01003881).}

\begin{abstract}
A tensor is a multi-dimensional array of complex numbers, and the Lohe tensor model is an aggregation model on the space of tensors with the same rank and size. It incorporates previously well-studied aggregation models on the space of low-rank tensors such as the Kuramoto model, Lohe sphere and matrix models as special cases. Due to its structural complexities in cubic interactions for the Lohe tensor model, explicit construction of solutions with specific structures looks daunting. Recently, we obtained completely separable states by associating rank-1 tensors. In this paper, we further investigate  another type of solutions, namely ``{\it quadratically separable states}" consisting of tensor products of matrices and their component rank-2 tensors are solutions to the double matrix model whose emergent dynamics can be studied using the same methodology of the Lohe matrix model.  
\end{abstract}

\maketitle \centerline{\date}


\section{Introduction} \label{sec:1}
\setcounter{equation}{0}
Collective behaviors often appear in   large population systems for weakly coupled oscillators or interacting units \cite{B-B, C-S, D-M-T-H} in diverse scientific disciplines including biology, social sciences, engineering with various space and time scales, for instance, colonies of bacteria \cite{P-P-L-P}, school of fish \cite{C-L-N-S, T-K-I-H}, flock of starlings \cite{Y-S-C-G},  pedestrian dynamics \cite{H-B-J-W}, opinion dynamics \cite{H-K}, power grid networks \cite{M-M-A-N}, etc. For a brief introduction to collective dynamics, we refer the reader to survey articles and book \cite{ A-B,C-H-L,H-K-P-Z, P-R-K}. Mathematical approach toward the understanding of collective motions has been established in literature  by Winfree \cite{Wi} and  Kuramoto \cite{Ku} in the 1970s, and by Vicsek \cite{Vi} in the 1990s. After their remarkable works, the aforementioned models have been extended in several directions, particularly in high-dimensional extension to Riemannian manifolds including the hypersurfaces \cite{A-M-D, H-K-R2} and the matrix Lie group \cite{H-K-R} which have attracted lots of interest thanks to its powerful application, for instance, nonconvex optimization. In this work, among high-dimensional models, we are concerned with the Lohe tensor model in \cite{H-P}. 

Next, we briefly discuss tensors and an aggregation model on the space of tensors, namely ``{\it the Lohe tensor model}".  
A rank-$m$ complex valued tensor can be represented as a  multi-dimensional array of complex numbers with multi-indices. The rank of a tensor is the number of indices, say a rank-$m$ tensor with size $d_1 \times \cdots \times d_m$ is an element of ${\mathbb C}^{d_1 \times \cdots \times d_m}$. For example, scalars, vectors and matrices correspond to rank-0, 1 and 2 tensors, respectively.  Let $T$ be a rank-$m$ tensor with a size $d_1 \times \cdots \times d_m$.  Then, we denote $(\alpha_1, \cdots, \alpha_m)$-th component of the tensor $T$ by $[T]_{\alpha_1 \cdots \alpha_m}$, and we set $\overline{T}$ by the rank-$m$ tensor whose components are the complex conjugate of the elements in $T$:
\[ [\overline{T}]_{\alpha_1 \cdots \alpha_m} :=\overline{[T]_{\alpha_1 \cdots \alpha_m}}. \]
Let ${\mathcal T}_m(\bbc; d_1 \times\cdots\times d_m)$ be the collection of all rank-$m$ tensors with size $d_1 \times\cdots\times d_m$. Then, it is a complex vector space.  Several well-known first-order aggregation models, for instance, the Kuramoto model \cite{Ku}, the swarm sphere model \cite{O1} and matrix models \cite{D, Lo09} can be regarded as aggregation models on ${\mathcal T}_0(\bbc; 1), {\mathcal T}_1(\bbr; d)$ and ${\mathcal T}_2(\bbc; d \times d)$, respectively. Let  $A_j$ be the skew-hermitian rank-$2m$ tensor with size $(d_1 \times\cdots\times d_m) \times (d_1 \times \cdots\times d_m)$. For simplicity, we introduce handy notation as follows: for $T \in {\mathcal T}_m(\bbc; d_1 \times \cdots\times d_m)$ and $A \in  {\mathcal T}_{2m}(\bbc; d_1 \times\cdots\times  d_m \times d_1 \times \cdots\times d_m)$, we set
\begin{align*}
\begin{aligned}
& [T]_{\alpha_{*}}:=[T]_{\alpha_{1}\alpha_{2}\cdots\alpha_{m}}, \quad [T]_{\alpha_{*0}}:=[T]_{\alpha_{10}\alpha_{20}\cdots\alpha_{m0}},  \quad  [T]_{\alpha_{*1}}:=[T]_{\alpha_{11}\alpha_{21}\cdots\alpha_{m1}}, \\
&  [T]_{\alpha_{*i_*}}:=[T]_{\alpha_{1i_1}\alpha_{2i_2}\cdots\alpha_{mi_m}}, \quad [T]_{\alpha_{*(1-i_*)}}:=[T]_{\alpha_{1(1-i_1)}\alpha_{2(1-i_2)}\cdots\alpha_{m(1-i_m)}}, \\
&  [A]_{\alpha_*\beta_*}:=[A]_{\alpha_{1}\alpha_{2}\cdots\alpha_{m}\beta_1\beta_2\cdots\beta_{m}}.
\end{aligned}
\end{align*}
Then, the Lohe tensor model in a component form reads as follows:
\begin{equation}
\begin{cases} \label{LT}
\displaystyle \dot{[T_j]}_{\alpha_{*0}} = [A_j]_{\alpha_{*0}\alpha_{*1}}[T_j]_{\alpha_{*1}} \\
\displaystyle \hspace{1.5cm} + \sum_{i_* \in \{0, 1\}^m}\kappa_{i_*} \Big([T_c]_{\alpha_{*i_*}}\bar{[T_j]}_{\alpha_{*1}}[T_j]_{\alpha_{*(1-i_*)}}-[T_j]_{\alpha_{*i_*}}\bar{[T_c]}_{\alpha_{*1}}[T_j]_{\alpha_{*(1-i_*)}} \Big), \\
\displaystyle  \bar{[A_j]}_{\alpha_{*0}\alpha_{*1}}=-[A_j]_{\alpha_{*1}\alpha_{*0}},
\end{cases}
\end{equation}
where $\kappa_{i_*}$'s are nonnegative coupling strengths. \newline

Before we discuss our main issues, we introduce a concept of a ``\textit{quadratically separable state"} for the Lohe tensor model \eqref{LT}. 
\begin{definition} \label{D1.1}
Let $\{T_i \}$ be a {\it quadratically separable state}   to \eqref{LT}, if it is decomposed as a tensor product of rank-2 tensors (or matrices):
\begin{equation*}
T_i = U_i^1 \otimes U_i^2 \otimes \cdots \otimes U_i^m, \quad U_i^k \in \bbc^{d_1^k \times d_2^k},\quad \|U_i^k\|_\tF = 1, \quad 1\leq i\leq N,\quad 1\leq k \leq m,
\end{equation*} 
where $\|\cdot\|_\tF$ is the Frobenius norm induced by Frobenius inner product: for matrices $A$ and $B$, 
\[ \langle A,B\rangle_\tF := \textup{tr}(A^\dg B), \quad \| A \|_\tF := \sqrt{\textup{tr}(A^\dg A)}. \] 
\end{definition}

\noindent In this paper, we are interested in the following simple questions:
\begin{itemize}
\item
(Q1):~Are there  quadratically separable states  for the Lohe tensor model?
\vspace{0.2cm}
\item
(Q2):~If so, do they exhibit   collective behaviors under which circumstances?
\end{itemize}
Our main results deal with the raised two questions (Q1) and (Q2). More precisely, our main results of this paper can be summarized as follows. 

First, we introduce the {\it double matrix model} induced from the Lohe tensor model whose elements have rank-4 with a specific condition on natural frequencies $B_j$ and $C_j$ (see \eqref{C-11}): 
\begin{align} \label{A-3}
\begin{cases}
 \dot{U}_j=B_jU_j+\displaystyle\frac{\kappa_{1}}{N}\sum_{k=1}^N\left(
\langle V_j, V_k\rangle_\tF~U_kU_j^\dagger U_j
-\langle V_k, V_j\rangle_\tF~U_jU_k^\dagger U_j\right)\\
 \hspace{2cm}+\displaystyle\frac{\kappa_{2}}{N}\sum_{k=1}^N\left(
\langle V_j, V_k\rangle_\tF~U_jU_j^\dagger U_k 
-\langle V_k, V_j\rangle_\tF~U_jU_k^\dagger U_j\right),\\
 \dot{V}_j=C_jV_j+\displaystyle\frac{\kappa_{1}}{N}\sum_{k=1}^N\left(
\langle U_j, U_k\rangle_\tF~V_kV_j^\dagger V_j
-\langle U_k, U_j\rangle_\tF~V_jV_k^\dagger V_j\right)\\
 \hspace{2cm}+\displaystyle\frac{\kappa_{2}}{N}\sum_{k=1}^N\left(
\langle U_j, U_k\rangle_\tF~V_jV_j^\dagger V_k
-\langle U_k, U_j\rangle_\tF~V_jV_k^\dagger V_j\right),\\
\end{cases}
\end{align}
where $B_j\in\bbc^{d_1\times d_2\times d_1\times d_2}$ and $C_j\in \bbc^{d_3\times d_4\times d_3\times d_4}$  are skew-symmetric rank-4 tensors. For a solution $\{(U_i,V_i)\}$ to the double matrix model \eqref{A-3}, a special solution $T_i$ to \eqref{LT} can be represented as follows:
\begin{equation*}
T_i(t)\equiv U_i(t)\otimes V_i(t),\quad t>0.
\end{equation*}
Precisely, if $T_i$ is initially decomposed into the tensor product of two matrices $U_i$ and $V_i$, then its separability is propagated along the flow for all time. For details, we refer the reader to Section \ref{sec:3}. 

Second, we study emergent dynamics of the double matrix model by investigating several aggregation quantities:
\begin{align*}
&\mathcal D(\mathcal U(t)) := \max_{1\leq i,j\leq N} \|U_i(t) - U_j(t)\|_\tF,\quad  \mathcal  S(\mathcal U(t)):= \max_{1\leq i,j\leq N} |n- \langle U_i,U_j\rangle_\tF (t)|, \\
&\mathcal D(\mathcal V(t)) := \max_{1\leq i,j\leq N} \|V_i(t) - V_j(t)\|_\tF, \quad \mathcal S(\mathcal V(t)):= \max_{1\leq i,j\leq N} |m- \langle V_i,V_j\rangle_\tF (t)|. 
\end{align*}
For a homogeneous ensemble (i.e. $B_i = B_j$ and $C_i = C_j$ for all $i$ and $j$.) we show that system \eqref{A-3} exhibits complete aggregation in which all  relative distances for $\{U_i\}$ and $\{V_i\}$ tend to zero respectively (see Theorem \ref{T4.1}). On the other hand, for a heterogeneous ensemble ($B_i \neq B_j$ and $C_i \neq C_j$ in general)  complete aggregation (one-point collapse) would not be expected. Instead, our concern is dedicated to emergence of locked states in which relative distances converge to positive definite values (see Theorem \ref{T4.2}). For our analytical results, we need to assume that the size of unitary matrices satisfy $\min(n, m)>4\sqrt{\max(n, m)}$ that requires restriction on $n,m$.  In fact, this technical assumption on the sizes is mainly due to the fact that elements are complex-valued. Thus, when the unitary groups $\Un$ and  $\Um$ are replaced by the special orthogonal groups $\textbf{SO}(n)$ and $\textbf{SO}(m)$, such restriction on $n,m$ would be removed (see Theorem \ref{TC.1} and Theorem \ref{TC.2}). \newline

The rest of the paper is organized as follows. In Section \ref{sec:2}, we begin with previous results on the relation between the Lohe tensor model and the swarm double sphere model presented in \cite{H-K-P}. As a natural extension, we construct the double matrix model in Section \ref{sec:3} and study   existence and uniqueness of quadratically separable states. In Section \ref{sec:4}, we study emergent dynamics of the double matrix model for both homogeneous and heterogeneous ensembles. Next, the double matrix model is further generalized to the multi matrix model in Section \ref{sec:5}.  Finally, Section \ref{sec:6} is devoted to a brief summary of the paper and future work. In Appendix \ref{sec:app.A} and Appendix \ref{sec:app.B}, we provide proofs of Lemma \ref{L4.1} and Lemma \ref{L4.2}, respectively. In Appendix \ref{sec:app.C}, emergent dynamics of the double matrix model on $\textbf{SO}(n) \times \textbf{SO}(m)$ is provided.

\vspace{0.5cm}

 For simplicity of presentation, we use the following abbreviated jargons: 
\begin{itemize}
\item LT model: Lohe tensor model, $\quad$ LM model:~Lohe matrix model, \\
\item SDS model: swarm double sphere model,$\quad$ SMS model: swarm multi-sphere model, \\
\item DM model: double matrix model, $\quad$ DUM model: double unitary matrix model, \\
\item DSOM model:  double special orthogonal matrix model, \quad QSS: quadratically separable state, \\
\item MM model: multiple matrix model, \quad MUM model: multiple unitary matrix model.
\end{itemize}
\section{Preliminaries} \label{sec:2}
\setcounter{equation}{0}
In this section, we review how the SDS model \cite{H-K-P, Lo20} can be related to the LT model via completely separable states, and discuss   extension of the SDS model to the SMS model leading to the DM model.
\subsection{From the LT model to the  SDS model} \label{sec:2.1} 
In this subsection, we   briefly recall the relation between the LT model and the SDS model which was first observed in \cite{H-K-P}. In \cite{Lo20}, Lohe introduced a first-order aggregation model on the product of two unit spheres $(u_i,v_i) \in \bbs^{d_1-1} \times \bbs^{d_2-1}$:
\begin{align}\label{B-0} 
\begin{cases}
\dot{u}_i=\Omega_iu_i+\displaystyle\frac{\kappa}{N}\sum_{j=1}^N\langle v_i, v_j\rangle (u_j-\langle u_i, u_j\rangle u_i) ,\quad t>0, \\
\dot{v}_i=\Lambda_i v_i+\displaystyle\frac{\kappa}{N}\sum_{j=1}^N \langle u_i, u_j\rangle(v_j-\langle v_i, v_j\rangle v_i),\\
(u_i, v_i)(0)=(u_i^0, v_i^0) \in\bbs^{d_1-1} \times \bbs^{d_2-1},\quad1\leq i\leq N, 
\end{cases}
\end{align}
where $\Omega_i$ and $\Lambda_i$ are skew-symmetric matrices of sizes $d_1 \times d_1$ and $d_2 \times d_2$, respectively: 
\begin{equation*} \label{B-1}
\Omega_i^\top = -\Omega_i, \quad \Lambda_i^\top = -\Lambda_i, \quad 1 \leq i \leq N,
\end{equation*}
and $\kp$ denotes the (uniform) coupling strength. \newline

On the other hand, if we choose the following parameters:
\begin{equation*}
m=2,\quad \kp_{00} = \kp_{11} =0,\quad \kp_{01} = \kp_{10} = \kp,
\end{equation*}
system \eqref{LT} reduces to the generalized Lohe matrix model in \cite{H-P3}:
\begin{equation}\label{B-2} 
\begin{cases}
\displaystyle {\dot T}_i =A_i T_i +\kappa(T_cT_i^\dagger T_i -T_i T_c^\dagger T_i)+\kappa(T_iT_i^\dagger T_c-T_i T_c^\dagger T_i), \quad t > 0, \\
\displaystyle T_i(0)=T_i^0,\quad \|T_i^0\|_\tF =1,\quad  T_c : = \frac{1}{N} \sum_{k=1}^{N} T_k, \quad i = 1, \cdots, N.
\end{cases}
\end{equation}
Next, we present how models \eqref{B-1} and \eqref{B-2} can be viewed as equivalent systems under  {\it well-prepared} natural frequency tensors and initial data in the following proposition.

\begin{proposition} \cite{H-K-P}  \label{P2.1} 
Systems \eqref{B-1} and \eqref{B-2} are equivalent in the following sense.
 \begin{enumerate}
 \item
 Suppose $\{(u_i,v_i) \}$ is a solution to \eqref{B-0}. Then, rank-2 real tensors $T_i$ defined by $T_i  :=u_i \otimes v_i$ is  a solution   to \eqref{B-2} with initial data $T_i^0 = u_i^0 \otimes v_i^0$ and well-prepared free flow tensors $A_i$:
 \begin{equation} \label{B-3}
 A_i T_i:= \Omega_i T_i + T_i \Lambda_i^\top.
 \end{equation}
 \item
 Suppose $T_i$ is a solution to \eqref{B-2}-\eqref{B-3} with completely factorized initial data:
 \begin{equation*}  
 T_i^0 =: u_i^0 \otimes v_i^0, \quad 1 \leq i, j \leq N,
 \end{equation*}
 for rank-1 real tensors $u_i^0 \in \bbs^{d_11}$ and $v_i^0 \in \bbs^{d_2-1}$. Then, there exists a pair of unit vectors $(u_i(t), v_i(t))$ such that 
 \[  T_i(t) = u_i(t) \otimes v_i(t), \quad t>0, \]
 where $(u_i,v_i)$ is a solution to \eqref{B-0} with initial data $(u_i,v_i)(0) = (u_i^0,v_i^0)$.
 \end{enumerate}
\end{proposition}
By applying the completely separability stated in Proposition \ref{P2.1}, emergent behaviors for \eqref{B-0} and those for \eqref{B-2} are exactly the same. Thus, it suffices to investigate the SDS model \eqref{B-0}. 
\begin{proposition}
\emph{\cite{H-K-P}}
Suppose the initial data $\{ (u_i^0, v_i^0) \}$ satisfy the following conditions:
\begin{equation*}
 \min_{1\leq i, j\leq N }\langle u_i^0, u_j^0\rangle>0,\quad  \min_{1\leq i, j\leq N}\langle v_i^0, v_j^0\rangle>0,
\end{equation*}
and let $\{(U,V)\}$ be a solution to system \eqref{B-0}.  Then, we have
\[
\lim_{t\to\infty} \max_{1\leq i,j\leq N}  |u_i(t)-u_j(t) |=0 \quad \textup{and} \quad \lim_{t\to\infty}  \max_{1\leq i,j\leq N}  |v_i(t)-v_j(t) |=0.
\]
\end{proposition}

\vspace{0.2cm}

Now, it is worthwhile mentioning that system \eqref{B-0} can be represented as a coupled gradient flow:
\begin{equation} \label{B-4}
\begin{cases}
\displaystyle {\dot u}_i =-\frac{N\kappa}{2} {\mathbb P}_{T_{u_i}\bbs^{d_1-1}} \Big( \nabla_{u_i}\mathcal{E}(U, V) \Big), \vspace{0.2cm}\\
\displaystyle {\dot v}_i =-\frac{N\kappa}{2}  {\mathbb P}_{T_{v_i}\bbs^{d_2-1}} \Big( \nabla_{v_i}\mathcal{E}(U, V)\Big),
\end{cases}
\end{equation}
where the projection operators $ {\mathbb P}_{T_{u_i}\bbs^{d_1-1}}$ and $ {\mathbb P}_{T_{v_i}\bbs^{d_2-1}}$ onto the tangent spaces of $\bbs^{d_1-1}$ and $\bbs^{d_2-1}$ at $u_i$ and $v_i$ are defined by the  formulae, respectively: for $w_1 \in \bbr^{d_1}$ and $w_2 \in \bbr^{d_2}$, 
\[
\begin{cases}
\displaystyle {\mathbb P}_{T_{u_i}\bbs^{d_1-1}} (w_1) := w_1 - \langle w_1, u_i \rangle u_i, \\
\displaystyle {\mathbb P}_{T_{v_i}\bbs^{d_2-1}} (w_2) := w_2 - \langle w_2, v_i \rangle v_i,
\end{cases}
\]
and the potential function $\mathcal E(U,V)$ is defined as 
\begin{equation} \label{B-5}
\mathcal E(U,V) := 1- \frac{1}{N^2} \sum_{i,j=1}^N \langle u_i,u_j\rangle \langle v_i,v_j\rangle.
\end{equation}
Thanks to the gradient flow formulation \eqref{B-4},  any solution to system \eqref{B-0} converges to an equilibrium as $t \to \infty$.


\subsection{From the SDS model to the SMS model} \label{sec:2.2} 
In this subsection, we  extend the SDS model \eqref{B-0} on the product of two unit spheres to an aggregation model on the product of multiple unit spheres, namely,  the SMS model. Note that the SDS model can be represented as a gradient flow with a potential function as can be seen in \eqref{B-4}--\eqref{B-5}. Thus, we first generalize the potential function \eqref{B-5} as follows: for $u_i^k \in \bbs^{d_k-1}, \quad i=1,\cdots,N, \quad k=1,\cdots,m,$ we set 
\begin{equation} \label{B-6}
\mathcal{E}(U^1, U^2, \cdots, U^m) :=1-\frac{1}{N^2}\sum_{i,j=1}^N\prod_{k=1}^m \langle u_i^k, u_j^k\rangle,\quad U^k :=\{u_1^k, u_2^k, \cdots, u_N^k\}.
\end{equation}
Using the same spirit for a gradient flow with the potential function \eqref{B-6}, we propose the SMS model as  follows.
\begin{align}\label{B-14}
\begin{cases}
\displaystyle {\dot u}_i^k =\frac{\kappa}{N}\sum_{j=1}^N\left(\prod_{l\neq k}\langle u_i^l, u_j^l\rangle\right) \Big(u_j^k- \langle u_i^k,u_j^k\rangle u_i^k\Big),\quad t>0, \\
u_i^k(0)=u_i^{k,0}\in\bbs^{d_k-1}  \qquad   i\in\{1, 2, \cdots, N\},\quad k\in\{1, 2, \cdots, m\}.
\end{cases}
\end{align}
As in Section \ref{sec:2.1}, we set rank-$m$ real tensor $T_i$:
\begin{equation*}
T_i := u_i^1 \otimes \cdots \otimes u_i^m,\quad i=1,\cdots,N. 
\end{equation*}
Then, it is easy to check that $T_i$ satisfies 
\begin{equation}\label{B-15}
[\dot{T}_p]_{\alpha_{*0}}=\frac{\kappa}{N}\sum_{k=1}^m \sum_{\ell=1}^N\left( [T_\ell]_{\alpha_{*i_*^k}}[\bar T_p]_{\alpha_{*1}}[T_p]_{\alpha_{*(1-i_*^k)}}-[T_p]_{\alpha_{*i_*^k}}[\bar T_\ell]_{\alpha_{*1}}[T_p]_{\alpha_{*(1-i_*^k)}}\right).
\end{equation}
It should be noted that \eqref{B-15} can be derivable from the Lohe tensor model \eqref{LT} with the following conditions:
\begin{align*}
\kappa_{i_*}=\begin{cases}
\kappa\quad&\text{when } i_*=i_*^k,\quad  1\leq k\leq m,\\
0&\text{otherwise},
\end{cases}
\quad\text{where}\quad i_*^k:=\underbrace{(0, \cdots, 0, 1,0, \cdots, 0)}_{\text{only $k^{th}$ index is $1$}}.
\end{align*}
Hence, systems \eqref{B-14} and \eqref{B-15} can be related in view of a completely separable state, and since the emergent dynamics of the LT model has been discussed in literature \cite{H-P, H-P2, H-P3}, we conclude that system \eqref{B-14} exhibits complete aggregation under suitable circumstances. 
\begin{proposition} \label{P2.3}
\emph{\cite{H-K-P}}
Suppose that initial data $T^0 = \{T_i^0\}$ are completely factorized as a tensor product of rank-1 real tensors:
\begin{align*}
\begin{aligned} 
& T_i^0=u_i^{1, 0}\otimes u_i^{2, 0}\otimes \cdots \otimes u_i^{m, 0}, \quad i = 1, \cdots, N, \\
& {\mathcal A}(U^{k,0}) := \min_{1 \leq i, j \leq N} \langle u^{k,0}_i, u^{k,0}_j \rangle  >0,\quad k=1,\cdots,m,
\end{aligned}
\end{align*}
and let $T=\{T_i \}$ be a solution to system \eqref{B-15} and $\{U^1,\cdots, U^d\}$ be a solution to system \eqref{B-14}. Then, the following assertions hold.
\begin{enumerate}
\item
$T_i=T_i(t)$ is completely separable in the sense that 
\[
T_i(t)=u_i^1(t)\otimes u_i^2(t)\otimes \cdots \otimes u_i^m (t),\quad t>0, \quad i = 1, \cdots, N.
\]
\item
The   solution  exhibits the complete aggregation:
\[ \lim_{t \to \infty} \max_{1\leq i, j \leq N} \|T_i(t) - T_j(t) \|_\tF  = 0. \]
\end{enumerate}
\end{proposition}

\begin{proof}
For a proof, we refer the reader to Theorem 6.2 and Proposition 7.1 in \cite{H-K-P}. 
\end{proof}

\section{Existence of rank-4 quadratically separable states} \label{sec:3} 
\setcounter{equation}{0}
In this section, we present  existence of the QSS for the LT model with rank-4 tensors and introduce the  DM model than can be induced from the LT model. 
\subsection{The DM model} \label{sec:3.1} 
In this subsection, we propose the DM model consisting of two generalized Lohe matrix model on the rectangular matrices with possibly different sizes: 
\[ U_j \in \bbc^{d_1\times d_2} \quad \mbox{and} \quad V_j \in \bbc^{d_3 \times d_4}, \quad j = 1, \cdots, N. \]
Below, we sketch our strategy how to derive the DM model from the LT model: 
\begin{itemize}
\item
Step A~(A homogeneous ensemble): we present  a DM model for a homogeneous ensemble (Section \ref{sec:3.1.1}). 
\vspace{0.2cm}
\item
Step B~(A heterogeneous ensemble):  by adding natural frequency tensors with suitable structure conditions, we derive the DM model  from the LT model (Section \ref{sec:3.1.2}). 
\end{itemize}
\subsubsection{A homogeneous ensemble}  \label{sec:3.1.1}
Let $T_j\in \bbc^{d_1\times d_2 \times d_3 \times d_4}$  be a rank-4 tensor that is a  solution to   \eqref{LT} with zero natural frequency tensors $A_j \equiv O$: 
\begin{equation} \label{C-1}
[\dot{T}_j]_{\alpha_{*0}}=\displaystyle\sum_{i_*\in\{0, 1\}^4}\left[ \frac{\kappa_{i_*}}{N}\sum_{k=1}^N\left([T_k]_{\alpha_{*i_*}}[\bar{T}_j]_{\alpha_{*1}}[T_j]_{\alpha_{*(1-i_*)}}-[T_j]_{\alpha_{*i_*}}[\bar{T}_k]_{\alpha_{*1}}[T_j]_{\alpha_{*(1-i_*)}}\right)\right ].
\end{equation}
For a given solution $T_j$ to \eqref{C-1}, we assume that there exist two matrices $U_j \in \bbc^{d_1 \times d_2}$  and  $V_j \in \bbc^{d_3 \times d_4}$ such that
\[
T_j = U_j  \otimes V_j,\quad [T_j(t)]_{\alpha\beta\gamma\delta}=[U_j(t)]_{\alpha\beta}[V_j(t)]_{\gamma\delta} \quad \mbox{in a component form}.
\]
Next, we rewrite cubic interaction terms in \eqref{C-1} in terms of $U_j$ and $V_j$. For this, we decompose the index vectors $i_*$ and $\alpha_{*i_*}$ as 
\[
i_* :=(i_1,i_2, i_3, i_4),\quad \alpha_{*i_*} := (\beta_{*j_*},\gamma_{*k_*}),\quad j_* :=(i_1, i_2),\quad k_* :=(i_3, i_4),
\]
where $j_*$ and $k_*$ correspond to the index vectors for $U_j$ and $V_j$, respectively. We now observe 
\begin{align*}
\begin{aligned}
&[T_i]_{\alpha_{*i_*}}[\bar{T}_j]_{\alpha_{*1}}[T_k]_{\alpha_{*(1-i_*)}} \\
& \hspace{0.5cm} =[T_i]_{(\beta_{*j_*},\gamma_{*k_*}}[\bar{T}_j]_{(\beta_{*1},\gamma_{*1})}[T_k]_{(\beta_{*(1-j_*)},\gamma_{*(1-k_*)})}\\
& \hspace{0.5cm} =[U_i]_{\beta_{*j_*}}[V_i]_{\gamma_{*k_*}}[\bar{U}_j]_{\beta_{*1}}[\bar{V}_j]_{\gamma_{*1}}[U_k]_{\beta_{*(1-j_*)}}[V_k]_{\gamma_{*(1-k_*)}}\\
& \hspace{0.5cm} =\left([U_i]_{\beta_{*j_*}}[\bar{U}_j]_{\beta_{*1}}[U_k]_{\beta_{*(1-j_*)}}\right)\left([V_i]_{\gamma_{*k_*}}[\bar{V}_j]_{\gamma_{*1}}[V_k]_{\gamma_{*(1-k_*)}}\right).
\end{aligned}
\end{align*}
By interchanging the roles of $j\leftrightarrow k$, the term inside of the summation in the right-hand side of \eqref{C-1} becomes
\begin{align*}
 &[T_k]_{\alpha_{*i_*}}[\bar{T}_j]_{\alpha_{*1}}[T_j]_{\alpha_{*(1-i_*)}}-[T_j]_{\alpha_{*i_*}}[\bar{T}_k]_{\alpha_{*1}}[T_j]_{\alpha_{*(1-i_*)}}\\
 &\hspace{1cm}=\left([U_k]_{\beta_{*j_*}}[\bar{U}_j]_{\beta_{*1}}[U_j]_{\beta_{*(1-j_*)}}\right)\left([V_k]_{\gamma_{*k_*}}[\bar{V}_j]_{\gamma_{*1}}[V_j]_{\gamma_{*(1-k_*)}}\right)\\
&\hspace{1.5cm}-\left([U_j]_{\beta_{*j_*}}[\bar{U}_k]_{\beta_{*1}}[U_j]_{\beta_{*(1-j_*)}}\right)\left([V_j]_{\gamma_{*k_*}}[\bar{V}_k]_{\gamma_{*1}}[V_j]_{\gamma_{*(1-k_*)}}\right).
\end{align*}
Since the left-hand side of \eqref{C-1} has the form of 
\begin{equation} \label{C-2} 
\dot T_j = \dot U_j \otimes V_j + U_j \otimes \dot V_j,
\end{equation}
one should impose either $\kp_* = (1,1)$ or $j_* = (1,1)$
to derive the restriction on $\kp_{i_*}$:
\begin{equation*} \label{C-5}
\kappa_{i_*}=0\quad \textup{for all $i_* \in \{0,1\}^4$ with $(i_1,i_2) \neq (1,1)$ and $(i_3,i_4) \neq (1,1)$.}
\end{equation*}
Then, the right-hand side of \eqref{C-1} further reduces to
\begin{align*}
\begin{aligned}\label{C-7}
&\displaystyle\sum_{i_*\in\{0, 1\}^4}\left[ \frac{\kappa_{i_*}}{N}\sum_{k=1}^N\left([T_k]_{\alpha_{*i_*}}[\bar{T}_j]_{\alpha_{*1}}[T_j]_{\alpha_{*(1-i_*)}}-[T_j]_{\alpha_{*i_*}}[\bar{T}_k]_{\alpha_{*1}}[T_j]_{\alpha_{*(1-i_*)}}\right)\right ] \\
&\hspace{1cm}= \displaystyle\sum_{j_*\in\{0, 1\}^2}\biggl[ \frac{\kappa_{(j_*, 1,1)}}{N}\sum_{k=1}^N\Big([T_k]_{\beta_{*j_*}\gamma_{*1}}[\bar{T}_j]_{\beta_{*1}\gamma_{*1}}[T_j]_{\beta_{*(1-j_*)}\gamma_{*0}} \\
&\hspace{5cm} -[T_j]_{\beta_{*j_*}\gamma_{*1}}[\bar{T}_k]_{\beta_{*1}\gamma_{*1}}[T_j]_{\beta_{*(1-j_*)}\gamma_{*0}}\Big)\biggr] \\
&\hspace{1.5cm}+\displaystyle\sum_{k_*\in\{0, 1\}^2}\biggl [ \frac{\kappa_{( 1,1, k_*)}}{N}\sum_{k=1}^N\Big([T_k]_{\beta_{*1}\gamma_{*k_*}}[\bar{T}_j]_{\beta_{*1}\gamma_{*1}}[T_j]_{\beta_{*0}\gamma_{*(1-k_*)}} \\
&\hspace{5cm}-[T_j]_{\beta_{*1}\gamma_{*k_*}}[\bar{T}_k]_{\beta_{*1}\gamma_{*1}}[T_j]_{\beta_{*0}\gamma_{*(1-k_*)}}\Big)\biggr ] \\
&\hspace{1cm}=[V_j]_{\gamma_{*0}}\displaystyle\sum_{j_*\in\{0, 1\}^2}\biggl[ \frac{\kappa_{(j_*, 1,1)}}{N}\sum_{k=1}^N\Big(\langle V_j, V_k\rangle_\tF [U_k]_{\beta_{*j_*}}[\bar{U}_j]_{\beta_{*1}}[U_j]_{\beta_{*(1-j_*)}}  \\
&\hspace{5cm}-\langle V_k, V_j\rangle_\tF[U_j]_{\beta_{*j_*}}[\bar{U}_k]_{\beta_{*1}}[U_j]_{\beta_{*(1-j_*)}}\Big)\biggr] \\
&\hspace{1.5cm}+[U_j]_{\beta_{*0}}\displaystyle\sum_{k_*\in\{0, 1\}^2}\biggl[ \frac{\kappa_{(1,1, k_*)}}{N}\sum_{k=1}^N\Big(
\langle U_j, U_k\rangle_\tF [V_k]_{\gamma_{*k_*}}[\bar{V}_j]_{\gamma_{*1}}[V_j]_{\gamma_{*(1-k_*)}} \\
&\hspace{5cm}-\langle U_k, U_j\rangle_\tF [V_j]_{\gamma_{*k_*}}[\bar{V}_k]_{\gamma_{*1}}[V_j]_{\gamma_{*(1-k_*)}}\Big)\biggr ].
\end{aligned}
\end{align*}
By comparing $\cdot \otimes V_j$ and $U_j \otimes \cdot$ in \eqref{C-2},  one has 
\begin{equation*}\label{C-8}
\begin{cases}
[\dot{U}_j]_{\beta_{*0}}=\displaystyle\sum_{j_*\in\{0, 1\}^2}\biggl(\frac{\kappa_{(j_*, 1,1)}}{N}\sum_{k=1}^N\Big(
\langle V_j, V_k\rangle_\tF [U_k]_{\beta_{*j_*}}[\bar{U}_j]_{\beta_{*1}}[U_j]_{\beta_{*(1-j_*)}} \\
\displaystyle \hspace{6cm} -\langle V_k, V_j\rangle_\tF[U_j]_{\beta_{*j_*}}[\bar{U}_k]_{\beta_{*1}}[U_j]_{\beta_{*(1-j_*)}}\Big)\biggr),\\
[\dot{V}_j]_{\gamma_{*0}}=\displaystyle\sum_{k_*\in\{0, 1\}^2}\biggl(\frac{\kappa_{(1,1, k_*)}}{N}\sum_{k=1}^N\Big(
\langle U_j, U_k\rangle_\tF [V_k]_{\gamma_{*k_*}}[\bar{V}_j]_{\gamma_{*1}}[V_j]_{\gamma_{*(1-k_*)}} \\
\displaystyle \hspace{6cm}-\langle U_k, U_j\rangle_\tF [V_j]_{\gamma_{*k_*}}[\bar{V}_k]_{\gamma_{*1}}[V_j]_{\gamma_{*(1-k_*)}}\Big)\biggr).\\
\end{cases}
\end{equation*}
If we choose the coupling strengths as for remaining $i_*$:
\begin{align}\label{C-9-1}
\kappa_{(0,1,1,1)}=\kappa_{(1, 1, 0, 1)}=\kappa_1,\quad \kappa_{(1, 0, 1, 1)}=\kappa_{(1, 1, 1, 0)}=\kappa_2,\quad \kappa_{i_*}=0,
\end{align}
we obtain the desired system for $(U_j, V_j)$: 
\begin{align}\label{C-10}
\begin{cases}
\dot{U}_j=\displaystyle\frac{\kappa_{1}}{N}\sum_{k=1}^N\left(
\langle V_j, V_k\rangle_\tF~U_kU_j^\dagger U_j
-\langle V_k, V_j\rangle_\tF~U_jU_k^\dagger U_j\right)
\\
\hspace{1cm}+\displaystyle\frac{\kappa_{2}}{N}\sum_{k=1}^N\left(
\langle V_j, V_k\rangle_\tF~U_jU_j^\dagger U_k 
-\langle V_k, V_j\rangle_\tF~U_jU_k^\dagger U_j\right),\\
\dot{V}_j=\displaystyle\frac{\kappa_{1}}{N}\sum_{k=1}^N\left(
\langle U_j, U_k\rangle_\tF~V_kV_j^\dagger V_j
-\langle U_k, U_j\rangle_\tF~V_jV_k^\dagger V_j\right) \\
\hspace{1cm}+\displaystyle\frac{\kappa_{2}}{N}\sum_{k=1}^N\left(
\langle U_j, U_k\rangle_\tF~V_jV_j^\dagger V_k
-\langle U_k, U_j\rangle_\tF~V_jV_k^\dagger V_j\right).\\
\end{cases}
\end{align}
\subsubsection{A heterogeneous ensemble} \label{sec:3.1.2}
Similar to several aggregation models such as the Lohe matrix model \cite{Lo09} and the Lohe tensor model \cite{H-P}, a natural candidate for heterogeneous (or non-identical) extension of  \eqref{C-10} would be the model \eqref{C-10} together with natural frequency tensors $B_j$ and $C_j$ whose ranks are 4.  Thus, the DM model for a heterogeneous ensemble reads as 
\begin{align}\label{C-10-1}
\begin{cases}
 \dot{U}_j=B_jU_j+\displaystyle\frac{\kappa_{1}}{N}\sum_{k=1}^N\left(
\langle V_j, V_k\rangle_\tF~U_kU_j^\dagger U_j
-\langle V_k, V_j\rangle_\tF~U_jU_k^\dagger U_j\right)\\
 \hspace{1cm}+\displaystyle\frac{\kappa_{2}}{N}\sum_{k=1}^N\left(
\langle V_j, V_k\rangle_\tF~U_jU_j^\dagger U_k 
-\langle V_k, V_j\rangle_\tF~U_jU_k^\dagger U_j\right),\\
 \dot{V}_j=C_jV_j+\displaystyle\frac{\kappa_{1}}{N}\sum_{k=1}^N\left(
\langle U_j, U_k\rangle_\tF~V_kV_j^\dagger V_j
-\langle U_k, U_j\rangle_\tF~V_jV_k^\dagger V_j\right)\\
 \hspace{1cm}+\displaystyle\frac{\kappa_{2}}{N}\sum_{k=1}^N\left(
\langle U_j, U_k\rangle_\tF~V_jV_j^\dagger V_k
-\langle U_k, U_j\rangle_\tF~V_jV_k^\dagger V_j\right),\\
\end{cases}
\end{align}
where $B_j\in\bbc^{d_1\times d_2\times d_1\times d_2}$ and $C_j\in \bbc^{d_3\times d_4\times d_3\times d_4}$  are rank-4 tensors satisfying skew-symmetric properties: for suitable indices, 
\begin{equation} \label{C-10-1-1}
[B_j]_{\alpha_1\beta_1\alpha_2\beta_2}=-[\bar{B}_j]_{\alpha_2\beta_2\alpha_1\beta_1},\quad [C_j]_{\gamma_1\delta_1\gamma_2\delta_2}=-[\bar{C}_j]_{\gamma_2\delta_2\gamma_1\delta_1}.
\end{equation}
Moreover, $U_jB_j$ and $C_jV_j$ can be understood as tensor contractions between rank-4 and rank-2 tensors:
 \[[B_jU_j]_{\alpha\beta}=[B_j]_{\alpha\beta\gamma\delta}[U_j]_{\gamma\delta}, \quad  [C_jV_j]_{\alpha\beta}=[C_j]_{\alpha\beta\gamma\delta}[V_j]_{\gamma\delta}.
 \]
 Next, we find a condition for $A_j$ in \eqref{LT} in terms of $B_j$ and $C_j$ in \eqref{C-10-1} so that models \eqref{LT} and \eqref{C-10-1} are equivalent. For this, it suffices to focus on the free flows by setting $\kp_1 = \kp_2=0$, i.e., 
 \[ \dot{U}_j=B_jU_j, \quad    \dot{V}_j=C_jV_j, \quad j = 1, \cdots, N. \]
 If we use the relations  above and  
 \[ (U_j \otimes V_j)' = (B_jU_j)\otimes V_j + U_j\otimes(C_j V_j), \]
then we can find 
\begin{equation} \label{C-10-2} 
 A_j(U_j\otimes V_j)=(B_j U_j)\otimes V_j+U_j\otimes(C_j V_j).
 \end{equation}
In addition, if we rewrite \eqref{C-10-2} as a component form, then it becomes
\begin{align} \label{C-10-3}
\begin{aligned}
0&=[A_j]_{\alpha_1\beta_1\gamma_1\delta_1\alpha_2\beta_2\gamma_2\delta_2}[U_j]_{\alpha_2\beta_2}[V_j]_{\gamma_2\delta_2} -[B_j]_{\alpha_1\beta_1\alpha_2\beta_2}[U_j]_{\alpha_2\beta_2}[V_j]_{\gamma_1\delta_1} \\
&\hspace{0.5cm}-[U_j]_{\alpha_1\beta_1}[C_j]_{\gamma_1\delta_1\gamma_2\delta_2}[V_j]_{\gamma_2\delta_2}\\
&= \Big([A_j]_{\alpha_1\beta_1\gamma_1\delta_1\alpha_2\beta_2\gamma_2\delta_2}-[B_j]_{\alpha_1\beta_1\alpha_2\beta_2}\delta_{\gamma_1\gamma_2}\delta_{\delta_1\delta_2}-[C_j]_{\gamma_1\delta_1\gamma_2\delta_2}\delta_{\alpha_1\alpha_2}\delta_{\beta_1\beta_2} \Big) \\
&\hspace{0.5cm}\times [U_j]_{\alpha_2\beta_2}[V_j]_{\gamma_2\delta_2}.
\end{aligned}
\end{align}
Since \eqref{C-10-3} holds for arbitrary $U_j$ and $V_j$, we can find an explicit relation between $A_j$ and $B_j$, $C_j$:
\begin{align}\label{C-11}
[A_j]_{\alpha_1\beta_1\gamma_1\delta_1\alpha_2\beta_2\gamma_2\delta_2}=[B_j]_{\alpha_1\beta_1\alpha_2\beta_2}\delta_{\gamma_1\gamma_2}\delta_{\delta_1\delta_2}+[C_j]_{\gamma_1\delta_1\gamma_2\delta_2}\delta_{\alpha_1\alpha_2}\delta_{\beta_1\beta_2}.
\end{align}
So far, we have derived the DM model \eqref{C-10-1} with \eqref{C-11} from the LT  model \eqref{LT}. However, if we use the same argument reversely and recall that system \eqref{C-10-1} admits a unique solution, then one can identify the LT model from the DM model. Thus, we can say that system \eqref{C-10-1}--\eqref{C-11} and  system \eqref{LT} are equivalent in some sense. Below, we summarize the argument above in the following proposition. 

 \begin{proposition} \label{P3.1} 
 \begin{enumerate}
 \item
 Suppose $\{(U_i,V_i) \}$ is a solution to \eqref{C-10-1}. Then, a rank-4 tensor  $T_i$ defined by $T_i  :=U_i \otimes V_i$ is a QSS to \eqref{LT} with a well-prepared free flow tensor $A_i$ satisfying \eqref{C-11}. 
 
 \vspace{0.2cm}
 
 \item
 Suppose a rank-4 tensor $T_i$ is a solution to \eqref{LT} with \eqref{C-11} and quadratically separable initial data:
 \begin{equation*} \label{C-3}
 T_i^0 =: U_i^0 \otimes V_i^0, \quad 1 \leq i \leq N,
 \end{equation*}
 for  rank-2 tensors $U_i^0 \in \bbc^{d_1\times d_2}$ and $V_i^0 \in \bbc^{d_3 \times d_4}$ with unit norms. Then, there exist two matrices with unit norms $U_i=U_i(t)$ and $V=V_i(t)$ such that 
 \[  T_i(t) = U_i(t) \otimes V_i(t), \quad t>0, \]
 where $(U_i,V_i)$ is a solution to \eqref{C-10-1} with $(U_i,V_i)(0) = (U_i^0,V_i^0)$.
 \end{enumerate}
  \end{proposition}

%
%
%

\subsection{Gradient flow formulation of the DM model}
In this subsection, we show that system \eqref{C-10-1} can be formulated as a gradient flow with a suitable analytical potential induced from the Lohe tensor model. Recall from \cite{H-P2} that the following functional can be associated with the LT model:
\begin{equation} \label{C-15}
\mathcal{V}(T)=1-\frac{1}{N^2}\sum_{i, j=1}^N\langle T_i, T_j\rangle_\tF.
\end{equation}
If we use the decomposition $T_i = U_i \otimes V_i$, then we find the corresponding functional for \eqref{C-10-1}:
\begin{equation} \label{C-15-1}
\mathcal{E}(U, V) :=1-\frac{1}{N^2}\sum_{i, j=1}^N\langle U_i\otimes V_i, U_j\otimes V_j\rangle_\tF=1-\frac{1}{N^2}\sum_{i, j=1}^N \langle U_i, U_j\rangle_\tF  \langle V_i, V_j\rangle_\tF.
\end{equation}
In the following lemma, we show that $\mathcal E(U,V)$ is non-increasing along the flow \eqref{C-10-1}.

\begin{lemma} \label{L3.1}
Let $\{(U_i,V_i)\}$ be a solution to \eqref{C-10-1} with $(B_j,C_j) = (O,O)$. Then, the functional $\mathcal E(U,V)$ is non-increasing in time:
\begin{align*}
 &\frac{\d}{\dt}\mathcal{E}(U,V) = -\frac{\kappa_1}{N}\sum_{j=1}^N  \left\| \frac{1}{N}\sum_{i=1}^N \Big( \langle V_j, V_i\rangle_\tF U_iU_j^\dagger-\langle V_i, V_j\rangle_\tF  U_jU_i^\dagger \Big) \right\|_\tF^2 \\
&\hspace{0.7cm}-\frac{\kappa_1}{N}\sum_{j=1}^N\left\| \frac{1}{N}\sum_{i=1}^N\left(\langle U_j, U_i\rangle_\tF V_iV_j^\dagger-\langle U_i, U_j\rangle_\tF  V_jV_i^\dagger\right)\right\|_\tF^2 \\
&\hspace{0.7cm}-\frac{\kappa_2}{N}\sum_{j=1}^N \left\| \frac{1}{N}\sum_{i=1}^N \Big( \langle V_j, V_i\rangle_\tF U_j^\dagger U_i-\langle V_i, V_j\rangle_\tF U_i^\dagger U_j \Big) \right\|_\tF^2 \\
&\hspace{0.7cm}-\frac{\kappa_2}{N}\sum_{j=1}^N\left\|\frac{1}{N}\sum_{i=1}^N\left(\langle U_j, U_i\rangle_\tF V_j^\dagger V_i-\langle U_i, U_j\rangle_\tF V_i^\dagger V_j\right)\right\|_\tF^2 .
\end{align*}
\end{lemma}

\begin{proof}
It follows from Proposition 4.1 in \cite{H-P2} that $\mathcal V(T)$ in \eqref{C-15} for the LT model \eqref{LT} satisfies 
\begin{equation}\label{C-12}
\frac{\d}{\dt}\mathcal{V}(T) =-\frac{1}{N}\sum_{j=1}^N\sum_{i_*\in\{0, 1\}^4}\kappa_{i_*}\|M^{i_*}(T_c)M^{i_*}(T_j)^\dagger-M^{i_*}(T_j)M^{i_*}(T_c)^\dagger\|_\tF^2.
\end{equation}
Here, we recall from \cite[Definition 3.1]{H-P2} that for a rank-4 tensor $T$,  $\mathcal M^{i_*}(T)$ is defined as a rank-2 tensor reshaped from $T$ or it can be understood as a bijective linear map which conserves the Frobenius norm. Since we only consider four types among $i_*\in \{ 0,1 \}^4$ as in \eqref{C-9-1}, we set  
\begin{equation} \label{C-13}
\kappa_{(0,1,1,1)}=\kappa_{(1, 1, 0, 1)}=\kappa_1,\quad \kappa_{(1, 0, 1, 1)}=\kappa_{(1, 1, 1, 0)}=\kappa_2,\quad \kappa_{i_*}=0\quad\text{for other }i_*.
\end{equation}
On the other hand, in \eqref{C-12}, we are concerned with the term $M^{i_*}(T_i)M^{i_*}(T_j)^\dagger$ for $T_i = U_i \otimes V_i$ and $i_*$ in \eqref{C-13}, for instance, if $i_*=(0, 1, 1, 1)$, then we have 
\begin{align*}
[M^{i_*}(T_i)M^{i_*}(T_j)^\dagger]_{\alpha_{20}\alpha_{21}}&=[T_i]_{\alpha_{11}\alpha_{20}\alpha_{31}\alpha_{41}}[\bar{T}_j]_{\alpha_{11}\alpha_{21}\alpha_{31}\alpha_{41}}\\
&=[U_i]_{\alpha_{11}\alpha_{20}}[V_i]_{\alpha_{31}\alpha_{41}}[\bar{U}_j]_{\alpha_{11}\alpha_{21}}[\bar{V}_j]_{\alpha_{31}\alpha_{41}}=\langle V_j, V_i\rangle_\tF  [U_j^\dagger U_i]_{\alpha_{21}\alpha_{20}}.
\end{align*}
If we repeat similar calculations as above, we obtain the desired dissipative estimate. 
\end{proof}

\subsection{The  DM model} \label{sec:3.3}
In this subsection, we further reduce the DM model \eqref{C-10-1} to the model defined on the product of two unitary matrices $\Un \times \Um$ which we call as the DUM model. By the construction of model \eqref{C-10-1}, we know that the Frobenius norms of $U_j$ and $V_j$ are conserved. In addition to the conservation of the Frobenius norms, we show that unitarity is also preserved, when rectangular matrices are replaced by square matrices.

\begin{lemma}\label{L3.2}
Suppose that the system parameters and initial data satisfy 
\[
d_1=d_2=n,\quad d_3=d_4=m,\quad U_j^0\in\mathbf{U}(n),\quad V_j^0\in\mathbf{U}(m),
\]
and let $\{(U_j, V_j)\}$ be a solution to \eqref{C-10-1}. Then, one has 
\begin{align*}\label{C-14}
U_j(t)\in\mathbf{U}(n),\quad V_j(t)\in\mathbf{U}(m),\quad t>0,\quad j=1,\cdots,N.
\end{align*}
\end{lemma}

\begin{proof}
Since $U_j$ and $V_j$ have the same structure, we focus only on $U_j$. We rewrite the dynamics of $U_j$ as a simpler form:
\begin{equation*}
\dot U_j  = B_j U_j + (P_j- P_j^\dg) U_j + U_j (Q_j-Q_j^\dg),
\end{equation*}
where $P_j$ and $Q_j$ are defined as
\begin{align*}
P_j:= \frac{\kp_1}{N} \sum_{k=1}^N \langle V_j,V_k\rangle_\tF U _k U_j^\dg,\quad Q_j := \frac{\kp_2}{N} \sum_{k=1}^N \langle V_j,V_k\rangle_\tF U_j^\dg U_k.
\end{align*}
Then, one has
\begin{align*}
\frac{\d}{\dt} (U_j U_j^\dg ) & = B_jU_j U_j^\dg - U_jU_j^\dg B_j +  (P_j- P_j^\dg) U_j U_j^\dg + U_j (Q_j- Q_j^\dg) U_j^\dg \\
&\hspace{0.5cm}- U_jU_j^\dg (P_j-P_j^\dg) - U_j ( Q_j- Q_j^\dg) U_j^\dg  \\
& = (B_j + P_j-P_j^\dg) U_j U_j^\dg - U_jU_j^\dg(B_j +P_j- P_j^\dg).
\end{align*}
By straightforward calculation, we find
\begin{equation}  \label{C-16}
 \frac{\d}{\dt} \|I_n-U_j U_j^\dg\|_\tF = 0. 
\end{equation}
Since we assume $U_j^0 \in \Un$, i.e., $I_n - U_jU_j^\dg=O$, the relation \eqref{C-16} yields the desired result. 
\end{proof}

Due to the unitarity, system \eqref{C-10-1} can further be reduced to the model on $\Un \times \Um$ by using the relations:
\begin{equation*}
U_j U_j^\dg = I_n = U_j^\dg U_j,\quad V_j V_j^\dg = I_m = V_j ^\dg V_j,\quad j=1,\cdots,N.
\end{equation*}
Thus, the DM model reads as
\begin{align*} 
\begin{cases}
\dot{U}_j=B_j U_j + \displaystyle\frac{\kappa}{N}\sum_{k=1}^N\left(
\langle V_j, V_k\rangle_\tF~U_k
-\langle V_k, V_j\rangle_\tF~U_jU_k^\dagger U_j\right),\\
\dot{V}_j=C_j V_j + \displaystyle\frac{\kappa}{N}\sum_{k=1}^N\left(
\langle U_j, U_k\rangle_\tF~V_k
-\langle U_k, U_j\rangle_\tF~V_jV_k^\dagger V_j\right),
\end{cases}
\end{align*}
where $\kappa=\kappa_1+\kappa_2$. \newline

Note that natural frequency tensors  $B_j$ and $C_j$ have rank-4 satisfying skew-symmetric properties as in \eqref{C-10-1-1}. In order to give a meaning of Hamiltonian, we associate two Hermitian matrices, namely, $H_j \in \bbc^{n\times n}$ and $G_j \in \bbc^{m \times m}$:
\[
[B_j]_{\alpha_1\beta_1\alpha_2\beta_2}=:[-\mathrm{i} H_j]_{\alpha_1\alpha_2}\delta_{\beta_1\beta_2},\quad [C_j]_{\gamma_1\delta_1\gamma_2\delta_2}=:[-\mathrm{i}G_j]_{\gamma_1\gamma_2}\delta_{\delta_1\delta_2}.
\]
Then, system \eqref{C-15} reduces to the model on $\Un\times \Um$:
\begin{align}  \label{C-17}
\begin{cases}
\dot{U}_j=-\mathrm{i} H_j U_j+\displaystyle\frac{\kappa}{N}\sum_{k=1}^N\left(
\langle V_j, V_k\rangle_\tF~U_k
-\langle V_k, V_j\rangle_\tF~U_jU_k^\dagger U_j\right),\\
\dot{V}_j=-\mathrm{i}G_j V_j+\displaystyle\frac{\kappa_{1}}{N}\sum_{k=1}^N\left(
\langle U_j, U_k\rangle_\tF~V_k
-\langle U_k, U_j\rangle_\tF~V_jV_k^\dagger V_j\right),
\end{cases}
\end{align}
where $H_j U_j$ and $G_j V_j$ are now usual matrix products.  Then as in Lemma \ref{L3.1}, system \eqref{C-17} also satisfies the dissipative energy estimate. Since the proof can be directly obtained from Lemma \ref{L3.1}, we omit it.

\begin{corollary}
Let $\{(U_i, V_i)\}$ be a solution to \eqref{C-17} with $H_j = G_j \equiv O$. Then the Lyapunov functional \eqref{C-15-1} satisfies 
\[
\frac{\d}{\dt}\mathcal{E}(U,V)=-\frac{\kappa}{N}\sum_{j=1}^N\left(\|\dot{U}_j\|_\tF^2+\|\dot{V}_j \|_\tF^2\right),\quad t>0. 
\]
\end{corollary}

\section{Emergent dynamics of rank-4 quadratically separable states} \label{sec:4}
\setcounter{equation}{0}
In this section, we study the emergent behavior of the QSS for the Lohe tensor model by investigating the dynamics of the DUM model which reads as follows: 
 \begin{equation}  \label{Z-0} 
\begin{cases}
\dot{U}_j=\displaystyle -\mi H_j U_j + \frac{\kappa}{N}\sum_{k=1}^N\left( \langle V_j, V_k\rangle_\tF~U_k
-\langle V_k, V_j\rangle_\tF~U_jU_k^\dagger U_j\right), \quad t >0, \\
\dot{V}_j=\displaystyle -\mi G_j   V_j + \frac{\kappa}{N}\sum_{k=1}^N\left(
\langle U_j, U_k\rangle_\tF~V_k
-\langle U_k, U_j\rangle_\tF~V_jV_k^\dagger V_j\right),\\
(U_j,V_j)(0) =(U_j^0, V_j^0) \in\mathbf{U}(n) \times  \mathbf{U}(m),  \quad 1\leq j \leq N, 
\end{cases}
\end{equation}
where $H_j \in  \bbc^{n\times n} $ and $G_j \in \bbc^{m \times m}$ are Hermitian matrices which play roles of natural frequencies for each oscillator. \newline

First, we recall several definitions of emergent behaviors for \eqref{Z-0}.
\begin{definition}  \label{D4.1}
\emph{\cite{H-R}}
Let $(\mathcal U,\mathcal V):=\{U_j,V_j\}_{j=1}^N$ be a solution to \eqref{Z-0}.  
\begin{enumerate}
\item
System \eqref{Z-0} exhibits complete aggregation if the following estimate holds:
\begin{equation} \label{Z-0-3} 
\lim_{t\to\infty} \max_{1\leq i,j\leq N} \Big( \|U_i(t) - U_j(t) \|_\tF + \|V_i(t) - V_j(t)\|_\tF \Big) =0.
\end{equation} 
\item
The state  $(\mathcal U,\mathcal V)$ tends to a phase-locked state if the following relation holds:
\begin{equation*}
\exists~ \lim_{t\to\infty} U_i(t) U_j(t)^\dg \quad \textup{and}\quad \exists~ \lim_{t\to\infty} V_i(t) V_j(t)^\dg.
\end{equation*}
\end{enumerate}
\end{definition}

\vspace{0.5cm}

In order to investigate emergent behaviors for \eqref{Z-0}, we denote the following quantities for notational simplicity:  
\begin{align} \label{Z-0-5}
\begin{aligned}
&X_{ij}:= U_i U_j^\dg, \quad S_{ij} := I_n - U_i U_j^\dg, \quad d_{ij} : = \langle U_i,U_j\rangle_\tF, \\
&Y_{ij} := V_i V_j^\dg,\quad T_{ij} := I_m - V_iV_j^\dg \quad c_{ij} := \langle V_i,V_j\rangle_\tF. 
\end{aligned}
\end{align}
It follows from simple observations that 
\begin{align*}
&\|U_j\|_\tF=n,\quad \|V_j\|_\tF = m,\quad |d_{ij}|\leq n,\quad |c_{ij}|\leq m, \\
&\|S_{ij}\|^2_\tF = \|U_i - U_j\|_\tF^2 = 2\textup{Re}( n-d_{ij})  ,\quad \|T_{ij}\|_\tF^2 = \|V_i - V_j\|_\tF^2 = 2\textup{Re}(m-c_{ij}). 
\end{align*}
Then, it is easy to see that 
\begin{equation*}
\|U_i - U_j\|_\tF \to 0 \Longleftrightarrow |n-d_{ij}| \to 0,\quad \|V_i - V_j\|_\tF \to 0 \Longleftrightarrow |m-c_{ij}| \to 0.
\end{equation*}
Thus, the complete aggregation in \eqref{Z-0-3} can be represented in terms of the quantities in \eqref{Z-0-5}:
\begin{align*}
&\lim_{t\to\infty} \max_{1\leq i,j\leq N } \Big( \|S_{ij}(t)\|_\tF + \|T_{ij}(t)\|_\tF \Big) = 0 \quad \textup{or} \quad \lim_{t\to\infty} \max_{1\leq i,j\leq N }  \Big(  |n-d_{ij}(t) | + |m-c_{ij}(t)| \Big) =0.
\end{align*}
In this regard, we define aggregation quantities: for $t>0$, 
\begin{align*}
&\mathcal D(\mathcal U(t)) := \max_{1\leq i,j\leq N} \|U_i(t) - U_j(t)\|_\tF,\quad  \mathcal  S(\mathcal U(t)):= \max_{1\leq i,j\leq N} |n- d_{ij}(t)|, \\
&\mathcal D(\mathcal V(t)) := \max_{1\leq i,j\leq N} \|V_i(t) - V_j(t)\|_\tF, \quad \mathcal S(\mathcal V(t)):= \max_{1\leq i,j\leq N} |m- c_{ij}(t)|. 
\end{align*}
Note that $\mathcal S(\mathcal U)$ has a second-order with respect to the state $\mathcal U$, whereas $\mathcal D(\mathcal U)$ has a first-order. \newline

In the following  two subsections, we present emergent dynamics for homogeneous and heterogeneous ensembles respectively. 
\subsection{Homogeneous ensemble} \label{sec:4.1} 
In this subsection,  we deal with the emergent dynamics of \eqref{Z-0} with  a homogeneous ensemble: 
\begin{equation*}
H_j \equiv H, \quad G_j \equiv G,\quad j=1,\cdots,N,
\end{equation*}
and by the solution splitting property, we may assume 
\begin{equation*}
H=O,\quad G=O. 
\end{equation*}
In this setting, system \eqref{C-10-1} becomes
\begin{align}\label{Z-1}
\begin{cases}
\dot{U}_j=\displaystyle\frac{\kappa}{N}\sum_{k=1}^N\left(
\langle V_j, V_k\rangle_\tF~U_k
-\langle V_k, V_j\rangle_\tF~U_jU_k^\dagger U_j\right), \quad t > 0, \\
\dot{V}_j=\displaystyle\frac{\kappa}{N}\sum_{k=1}^N\left(
\langle U_j, U_k\rangle_\tF~V_k
-\langle U_k, U_j\rangle_\tF~V_jV_k^\dagger V_j\right),\\
(U_j,V_j)(0) =(U_j^0, V_j^0) \in\mathbf{U}(n) \times  \mathbf{U}(m). 
\end{cases}
\end{align}
Without loss of generality, we may assume 
\begin{equation} \label{Z-8}
n\geq m.
\end{equation}
Our goal of this subsection is to find a sufficient condition under which 
\begin{equation*}
\lim_{t\to\infty}\mathcal L(t)=0,\quad  \mathcal L(t):=  \mathcal D(\mathcal U(t))  +\mathcal  D(\mathcal V(t)) + \mathcal S(\mathcal U(t) ) +\mathcal  S(\mathcal V(t)) ,
\end{equation*}
where $\mathcal L=\mathcal  L(t)$ is called as a total aggregation functional. \newline

In \cite{H-R}, it suffices to study the temporal evolutions of $\mathcal D(\mathcal U)$. However in our case, time evolutions of $\mathcal S(\mathcal U)$ as well as $\mathcal S(\mathcal V)$ are needed to achieve complete aggregation estimates. Below, we derive a differential inequality for $\mathcal L$.  
\begin{lemma} \label{L4.1} 
Let $\{(U_j,V_j)\}$ be a solution to \eqref{Z-1} with \eqref{Z-8}. Then, the total  aggregation functional $\mathcal L$ satisfies
\begin{align} \label{Z-9}
\begin{aligned}
\dot  {\mathcal L} \leq   -2\kp (m-4\sqrt n)\mathcal L +  \kp (4n+9)\mathcal L^2 +\kp \left( 2n + \frac 83\right)\mathcal L^3, \quad t>0.
\end{aligned}
\end{align}
\end{lemma}
\begin{proof}
Since a proof is lengthy, we provide it in Appendix \ref{sec:app.A}.
\end{proof}

We are now ready to provide a sufficient condition leading to the complete aggregation for \eqref{Z-1}. 

 \begin{theorem} \label{T4.1} 
Suppose that the system parameters and initial data satisfy 
\begin{align} \label{Z-10}
\begin{aligned}
&\textup{(i)}~~n\geq m >4\sqrt n. \\
&\textup{(ii)}~~ \mathcal L^0 <\alpha_{n,m} := \frac{-(12n+27) + \sqrt{ (12n+27)^2 + 48(m-4\sqrt n)(3n+4)  } }{4(3n+4)},
\end{aligned}
\end{align}
 and let $\{(U_j,V_j)\}$ be a solution to \eqref{Z-1}. Then, we have 
 \begin{equation*}
 \lim_{t\to\infty} \mathcal L(t) = 0.
 \end{equation*}
 Moreover, the convergence rate is at least exponential. In other words, system \eqref{Z-1} exhibits complete aggregation with an exponential convergence. 
 \end{theorem}
 \begin{proof}
Consider an auxiliary quadratic polynomial:
\begin{equation*}
f(s):= \left( 2n + \frac83\right) s^2 + (4n+9)s - 2(m-4\sqrt n).
\end{equation*}
Since $m>4\sqrt n$, the algebraic relation $f=0$ admits a unique positive root $\alpha_{n,m}$ defined in \eqref{Z-10}(ii).  Then, the relation \eqref{Z-9} is rewritten in terms of $f$:
\begin{equation*}
\frac{\d}{\dt} \mathcal L \leq \kp \mathcal L f(\mathcal L),\quad t>0.
\end{equation*}
Since initial data satisfy \eqref{Z-10}(ii), the desired result follows from dynamical systems theory.  
 \end{proof}
 \begin{remark}
\textup{(i)}~~ In \eqref{Z-10}(i), we have assumed that $n\geq m >4\sqrt n$ which imposes a restriction on the size of $U_i$ such as 
 \begin{equation*}
 n>16. 
 \end{equation*}
 It  also should be mentioned that such restriction arises, for instance,  when we estimate the terms $\mathcal I_{15}$ in \eqref{AA-9} and $\mathcal I_{25}$ in \eqref{AA-17}.  We indeed show in Appendix \ref{sec:app.C} that this technical assumption would be removed.  \newline
 
 \noindent \textup{(ii)}~~Since $m\leq n$, we have
 \begin{equation*}
 \limsup_{n\to\infty} \alpha_{n,m}  \leq \lim_{n\to \infty}  \frac{-(12n+27) + \sqrt{ (12n+27)^2 + 48(n-4\sqrt n)(3n+4)  } }{4(3n+4)} = \frac12.
 \end{equation*}
 \end{remark} 
 \subsection{Heterogeneous ensemble} \label{sec:4.2} 
In this subsection, we study the case of heterogeneous Hamiltonians in which $H_j$ and $G_j$ in \eqref{Z-0} are given to be different in general. 
In order to establish the emergence of the phase-locked state, we will follow a strategy developed in \cite{H-R}. For any two solutions $\{U_j,V_j\}$ and $\{\tilde U_j, \tilde V_j\}$ to \eqref{Z-0}, we define the diameters measuring the dissimilarity of two configurations: 
\begin{align} \label{Z-50-2}
\begin{aligned}
&d(U,\tilde U)(t) := \max_{1\leq i,j\leq N} \|U_i(t) U_j^\dg(t) - \tilde U_i(t) \tilde U_j^\dg(t)\|_\tF,  \\
&d(V,\tilde V)(t) := \max_{1\leq i,j\leq N} \|V_i(t) V_j^\dg(t) - \tilde V_i(t) \tilde V_j^\dg(t)\|_\tF. 
\end{aligned}
\end{align} 
Then, we will show that the diameters above converge to zero:
\begin{equation}  \label{Z-50-3}
\lim_{t\to\infty} \Big( d(U,\tilde U)(t)  + d(V,\tilde V)(t) \Big) = 0.
\end{equation}
As a next step, since our system is autonomous, for any $T>0$, we choose $\tilde U_j$ and $\tilde V_j$ as 
\begin{equation*}
\tilde U_j(t) = U_j(t+T),\quad \tilde V_j(t) = V_j(t+T).
\end{equation*} 
By discretizing the time $t\in \bbr_+$ as $n\in \bbz_+$ and setting $T=m \in \bbz_+$, we conclude that $\{U_i(n)U_j^\dg(n)\}_{n\in \bbz_+}$ and  $\{V_i(n)V_j^\dg(n)\}_{n\in \bbz_+}$ are indeed Cauchy sequences in the complete spaces $\Un$ and $\Um$, respectively. Hence, there exist two constant unitary matrices $U_{ij}^\infty \in \Un$ and $V_{ij}^\infty \in \Um$ such that
\begin{equation*}
\lim_{t\to\infty} \| U_i(t)U_j^\dg(t) - U_{ij}^\infty \|_\tF = 0,\quad \lim_{t\to\infty} \|V_i(t) V_j^\dg(t) - V_{ij}^\infty\|_\tF = 0.
\end{equation*}
Hence, we aim to find a sufficient condition under which  \eqref{Z-50-3} holds. To show \eqref{Z-50-3}, we associate another diameters measuring the difference between two solution configurations $\{U_j,V_j\}$ and $\{\tilde U_j, \tilde V_j\}$:
\begin{align} \label{Z-50-6}
\begin{aligned}
&\mathcal S(U,\tilde U)(t) := \max_{1\leq i,j\leq N} | \langle U_i,U_j\rangle(t) - \langle \tilde U_i,\tilde U_j\rangle (t)|, \\
&\mathcal S(V,\tilde V)(t) := \max_{1\leq i,j\leq N} | \langle V_i,V_j\rangle(t) - \langle \tilde V_i,\tilde V_j\rangle (t)|.
\end{aligned}
\end{align}
Note that 
\begin{equation*}
\mathcal S(U,\tilde U) \leq \sqrt n d(U,\tilde U),\quad \mathcal S(V,\tilde V) \leq \sqrt m d(V,\tilde V).
\end{equation*}
To this end, our goal of this subsection is to find a sufficient framework leading to 
\begin{equation} \label{Z-52}
\lim_{t\to\infty} \mathcal F(t) = 0,\quad \mathcal F(t) := d(U,\tilde U)(t) + d(V,\tilde V)(t) + \mathcal S(U,\tilde U)(t) + \mathcal S(V,\tilde V)(t).
\end{equation}
Below, we derive a differential inequality for $\mathcal F$ in \eqref{Z-52}. 
  
  \begin{lemma} \label{L4.2} 
  Let $\{( U_i, V_i)\}$ and $\{( \tilde U_i,\tilde  V_i)\}$ be any two solutions to \eqref{Z-0} with \eqref{Z-8}, respectively. Then, the aggregation functional $\mathcal F$ satisfies
  \begin{equation} \label{Z-55}
\frac{\d}{\dt} { \mathcal F} \leq  -\kp \left( 2   m - 8  \sqrt n - \frac{ \max\{ \mathcal D(\mathcal H), \mathcal D(\mathcal G)\} }{\kp} \right) \mathcal F   + \kp ( 4n + 22) \mathcal L \mathcal F + 20\kp \mathcal L^2 \mathcal F,
  \end{equation}
  where $\mathcal D(\mathcal H)$ and $\mathcal D(\mathcal G)$ are defined as
  \begin{equation*}
  \mathcal D(\mathcal H):= \max_{1\leq i,j\leq N } \|H_i - H_j\|_\infty, \quad   \mathcal D(\mathcal G):= \max_{1\leq i,j\leq N } \|G_i - G_j\|_\infty.
  \end{equation*}
  \end{lemma}
  \begin{proof}
  We postpone its proof in Appendix \ref{sec:app.B}. 
  \end{proof} 
  
  In what follows,  using the differential inequality \eqref{Z-55}, we provide a sufficient condition leading to the phase-locked state. First,  in order to make  a leading coefficient of \eqref{Z-55}  negative, we assume that a coupling strength $\kp$ is sufficiently large: 
  \begin{equation*}
  \kp > \frac{ \max\{ \mathcal D(\mathcal H), \mathcal D(\mathcal G)\} }{2(m -4\sqrt n)}.
  \end{equation*}
For a handy notation, we denote
\begin{equation*}
\Lambda := 2(m-4\sqrt n) - \frac{ \max\{ \mathcal D(\mathcal H), \mathcal D(\mathcal G)\} }{\kp}.
\end{equation*}

Next, we show that we can make the total aggregation functional $\mathcal L$ small as we wish by controlling the coupling strength $\kp$. 

\begin{proposition} \label{P4.1} 
Suppose system parameters and initial data satisfy 
\begin{align}  \label{Z-60}
\begin{aligned}
&\textup{(i)}~~ n\geq m,\quad \mathcal D(\mathcal H) \geq \mathcal D(\mathcal G), \qquad \textup{(ii)}~~ \kp\geq\kp_\textup{c}, \quad \mathcal L^0 \ll \nu_2, 
\end{aligned}
\end{align}
where $\kp_\textup{c}$ and $\nu_2$ are specified later in \eqref{Z-70} and \eqref{Z-71}, respectively, and let $\{(U_j,V_j)\}$ be a solution to \eqref{Z-0}. Then, we have
\begin{equation*}
\lim_{\kp \to\infty} \limsup_{t\to\infty} \mathcal L(t) = 0.
\end{equation*}
\end{proposition}

\begin{proof}
Since we assumed \eqref{Z-60}(i), it follows from \eqref{Z-9} that
\begin{equation} \label{Z-62}
\dot{\mathcal L} \leq -2\kp (m-4\sqrt n)\mathcal L + \kp(4n+9)\mathcal L^2 + \kp\left( 2n+\frac83 \right) \mathcal L^3 + 2(1+3\sqrt n)\mathcal D(\mathcal H).
\end{equation}
Now, we introduce an auxiliary cubic polynomial:
\begin{equation} \label{Z-63}
g(s) := 2(m-4\sqrt n ) s -(4n+9) s^2 - \left( 2n+ \frac83\right) s^3.
\end{equation}
Then, \eqref{Z-62} can be rewritten as
\begin{equation*}
\dot {\mathcal L} \leq \kp \left( \frac{ 2(1+3\sqrt n) \mathcal D(\mathcal H) }{\kp}- g(\mathcal L) \right).
\end{equation*}
By investigating roots of the polynomial $g$ in \eqref{Z-63}, we deduce that for a sufficient large $\kp$, the polynomial $g$ admits  one negative root, say, $\nu_0<0$, and two positive roots, say, $0<\nu_1<\nu_2$ with  continuous dependence of $\kp$, i.e., 
\begin{equation} \label{Z-65}
\lim_{\kp\to\infty} \nu_1(\kp) = 0,\quad \lim_{\kp\to\infty} \nu_2(\kp) = \alpha_*,\quad \textup{$\alpha_*$: a unique positive root of $g$}. 
\end{equation}
Since we assume \eqref{Z-60}(ii), there exists a finite entrance time $T_*>0$ such that
\begin{equation} \label{Z-66}
\mathcal L(t) <\nu_1,\quad t>T_*. 
\end{equation}
Finally, we combine \eqref{Z-65} and \eqref{Z-66} to obtain the desired estimate:
\begin{equation*}
\lim_{\kp \to\infty} \limsup_{t\to\infty} \mathcal L(t) = 0.
\end{equation*}
\end{proof}

\begin{remark}
For an explicit value for $\kp_\textup{c}$ in \eqref{Z-60}\textup{(ii)}, by simple calculus, we know that the polynomial $g=g(s)$ in \eqref{Z-63} admits a global maximum in $\bbr_+$ at $s=s_*$:
\begin{equation*}
s_* := \frac{-(4n+9) + \sqrt{ (4n+9)^2 + 2(6n+8)(m-4\sqrt n)  }  }{6n+8}.
\end{equation*} 
Thus, in order to guarantee the existence of a positive root for $g$, one should impose
\begin{equation} \label{Z-70}
g(s_*) > \frac{  2(1+3\sqrt n)\mathcal D(\mathcal H)}{\kp},\quad \textup{i.e.,} \quad  \kp >  \frac{  2(1+3\sqrt n)\mathcal D(\mathcal H)}{g(s_*)}=:\kp_\textup{c}.
\end{equation}
For this $\kp_\textup{c}$, $\nu_2$ in \eqref{Z-60}(ii) can be explicitly determined as the largest positive root of 
\begin{equation} \label{Z-71} 
g(s) = \frac{  2(1+3\sqrt n)\mathcal D(\mathcal H)}{\kp} \quad \textup{when $\kp>\kp_\textup{c}$}.
\end{equation} 
\end{remark}

\vspace{0.2cm}

It follows from Proposition \ref{P4.1} that under the assumption on the smallness of the initial data $\mathcal L^0$ and largeness of the coupling strength $\kp$, we have 
\begin{equation*}
\lim_{\kp\to\infty} \limsup_{t\to\infty} \mathcal L(t) =0. 
\end{equation*}
Hence, there exists a (large) coupling strength  $\kp_\textup{p}$ larger than $ \kp_\textup{c}$ such that for $t\gg1$, 
\begin{equation*}
(4n+22)\mathcal L (t)+ 20\mathcal L(t)^2 <\frac{\Lambda}{2},\quad \Lambda = 2(m-4\sqrt n) - \frac{ \max\{\mathcal D(\mathcal H),\mathcal D(\mathcal G) \}}{\kp}.
\end{equation*}
Then, \eqref{Z-55} becomes 
\begin{align*}
\frac\d\dt \mathcal F &\leq - \Lambda \kp \mathcal F + \kp(4n+22)\mathcal L \mathcal F + 20 \kp \mathcal L^2 \mathcal F  \leq  -\Lambda \kp \mathcal F + \frac{\Lambda\kp}{2} \mathcal F = -\frac{\Lambda\kp}{2} \mathcal F .
\end{align*}
This yields a desired exponential decay for $\mathcal F$. The argument above can be stated and shown as follows.

\begin{theorem} \label{T4.2} 
Suppose system parameters and initial data satisfy 
\begin{align*}
\begin{aligned}
&\textup{(i)}~~n\geq m,\quad \mathcal D(\mathcal H)\geq \mathcal D(\mathcal G), \qquad &\textup{(ii)}~~\kp > \kp_\textup{p} > \kp_\textup{c},\quad \max\{ \mathcal L^0, \tilde{\mathcal L}^0\} \leq \nu_2,
\end{aligned}
\end{align*}
and let $\{(U_i,V_i)\}$ and $\{(\tilde U_i, \tilde V_i)\}$ be any two solutions to \eqref{Z-0}, respectively.  Then, the following assertions hold. 
\begin{enumerate}
\item
The  functional $\mathcal F$ converges to zero with an exponential rate. 
\item
The normalized velocities $\mi \dot U_j U_j^\dg$ and $\mi \dot V_j V_j^\dg$ synchronize:
\[ \lim_{t \to \infty}  \|\mi \dot U_j U_j^\dg - \mi \dot{\tilde U}_j \tilde U_j^\dg \|_\tF = 0, \quad  \lim_{t \to \infty} \|\mi \dot V_j V_j^\dg - \mi \dot{\tilde V}_j \tilde V_j^\dg \|_\tF = 0. \]
\item
There exist unitary matrices $X_\infty \in \Un$ and $Y_\infty \in \Um$  such that
\begin{align*}
& \lim_{t\to\infty} U_i^\dg(t) \tilde U_i(t) = X_\infty, \quad \lim_{t\to\infty} \|\tilde U_i(t) - U_i(t) X_\infty\|_\tF =0, \\
& \lim_{t\to\infty} V_i^\dg(t) \tilde V_i(t) = Y_\infty,\quad \lim_{t\to\infty} \|\tilde V_i(t) - V_i(t) Y_\infty\|_\tF =0.
\end{align*}
\item
Asymptotic phase-locking emerge: for any  indices $i$ and $j$, 
\begin{equation*}
\exists~\lim_{t\to\infty} U_i(t) U_j^\dg (t) \quad \textup{and}\quad \exists~\lim_{t\to\infty} U_i(t) U_j^\dg (t).
\end{equation*}
Moreover, there exist phase-locked state $\mathcal X^\infty := \{X_i^\infty\}_{i=1}^N$ and $\mathcal Y^\infty := \{Y_i^\infty\}_{i=1}^N$, and unitary matrices $P \in \Un$ and $Q \in \Um$  such that
\begin{align*}
&\lim_{t\to\infty} \| U_i(t) - X_i^\infty P\|_\tF = 0 = \lim_{t\to \infty} \|V_i(t) - Y_i^\infty Q\|_\tF = 0, \quad \max \{ \mathcal D(\mathcal X^\infty), \mathcal D(\mathcal Y^\infty) \} < \nu_2. 
\end{align*}
\end{enumerate}
\end{theorem}
\begin{proof}
(i)~It follows from \eqref{Z-55} in Proposition \ref{P4.1} that  $\mathcal F$ satisfies 
\begin{equation} \label{Z-75}
\frac{\d}{\dt} { \mathcal F} \leq  -\kp \left( 2   m - 8  \sqrt n - \frac{ \max\{ \mathcal D(\mathcal H), \mathcal D(\mathcal G)\} }{\kp} \right) \mathcal F   + \kp ( 4n + 22) \mathcal L \mathcal F + 20\kp \mathcal L^2 \mathcal F.
\end{equation}
Since $\mathcal L$ can be sufficiently small under \eqref{Z-60}, i.e.,  
\begin{equation*}
\mathcal L(t) <\nu_1,\quad t>T_*,
\end{equation*}
we choose $\kp_\textup{p} >\kp_\textup{c}$ sufficiently large such that 
\begin{equation*}
(4n+22)\mathcal L + 20\mathcal L^2 <\frac{\Lambda}{2},\quad \Lambda = 2(m-4\sqrt n) - \frac{ \max\{\mathcal D(\mathcal H),\mathcal D(\mathcal G) \}}{\kp}.
\end{equation*}
Thus, \eqref{Z-75} becomes
\begin{equation} \label{Z-77}
\frac\d\dt \mathcal F \leq - \frac{\Lambda\kp}{2} \mathcal F,\quad t>T_*,
\end{equation}
and the relation \eqref{Z-77} yields the desired exponential convergence of $\mathcal F$ toward zero. \newline 

\noindent (ii) For the second assertion, we claim:
\begin{equation} \label{Z-78}
\|\mi \dot U_j U_j^\dg - \mi \dot{\tilde U}_j \tilde U_j^\dg \|_\tF \leq   2\kp(m+n) \mathcal F, \qquad \|\mi \dot V_j V_j^\dg - \mi \dot{\tilde V}_j \tilde V_j^\dg \|_\tF \leq   2\kp(m+n) \mathcal F.
\end{equation}
Once the relations \eqref{Z-78} hold, it follows from the first assertion to derive the second assertion. Note that 
\begin{align} \label{Z-79}
\begin{aligned}
&\|\mi \dot U_j U_j^\dg - \mi \dot{\tilde U}_j \tilde U_j^\dg \|_\tF  = \left \|  \frac{\kp}{N}\sum_{k=1}^N \Big( \langle V_j,V_k\rangle_\tF U_k U_j^\dg - \langle V_k,V_j\rangle_\tF U_jU_k^\dg  \Big)  \right. \\
& \hspace{5cm}  \left. - \frac{\kp}{N}\sum_{k=1}^N \Big( \langle \tilde V_j,\tilde V_k\rangle_\tF \tilde U_k \tilde U_j^\dg - \langle \tilde V_k,\tilde V_j\rangle_\tF \tilde U_j\tilde U_k^\dg  \Big)  \right\|_\tF.
\end{aligned}
\end{align}
By algebraic manipulations, one has 
\begin{align*}
&\| \langle V_j,V_k\rangle_\tF U_k U_j^\dg - \langle \tilde V_j,\tilde V_k\rangle_\tF \tilde U_k \tilde U_j^\dg \|_\tF \\
& \hspace{0.5cm} \leq \|\langle V_j,V_k\rangle_\tF (U_kU_j^\dg - \tilde U_k \tilde U_j^\dg )\|_\tF + \| ( \langle V_j,V_k\rangle_\tF - \langle \tilde V_j,\tilde V_k\rangle_\tF ) \tilde U_k\tilde U_j^\dg \|_\tF\leq (m+n) \mathcal F.
\end{align*} 
Thus, the relation \eqref{Z-79} yields
\begin{equation*}
\|\mi \dot U_j U_j^\dg - \mi \dot{\tilde U}_j \tilde U_j^\dg \|_\tF \leq    2\kp(m+n) \mathcal F .
\end{equation*}
This shows the desired synchronization of the normalized velocities and the exactly same argument is applied to $\mi \dot V_j V_j^\dg$. \newline

\noindent (iii) Since system \eqref{Z-0} is autonomous, we directly use Theorem 2(3) of \cite{H-R}. Hence, we briefly sketch a proof. We observe 
\begin{align*}
\left\|   \frac{\d}{\d s} (U_i^\dg \tilde U_i ) \right\|_\tF = \| \dot U_i U_i^\dg - \dot{\tilde U}_i \tilde U_i^\dg \|_\tF  \leq 2\kp(m+n)\mathcal F.
\end{align*}
Since $\mathcal F$ tends to zero exponentially, 
\begin{equation} \label{Z-82}
\lim_{t\to\infty} (U_i^\dg(t) \tilde U_i(t)) = U_i^{0,\dg}\tilde U_i^0 +\int_0^\infty  \frac{\d}{\d s} ( U_i^\dg(s) \tilde U_i(s)) \d s.
\end{equation}
In addition,  since $d(U,\tilde U)$ converges to zero, one deduces that the right-hand side of \eqref{Z-82} does not depend on the index $i$. This establishes the desired assertion. \newline

\noindent (iv) We first show the existence of the asymptotic limit of $U_iU_j^\dg$ and $V_iV_j^\dg$. For any $T>0$, since system \eqref{Z-0} is autonomous, we choose $\tilde U_i$ and $\tilde V_i$ as 
\begin{equation*}
\tilde U_i(t)  = U_i(t+T),\quad V_i(t) = V_i(t+T).
\end{equation*}
If we discretize the time $t\in\bbr_+$ and $n \in \bbz_+$ and choose $T= m \in \bbz_+$, we use the convergence of $d(U,\tilde U)$ and $d(V,\tilde V)$ to conclude that  $\{U_i(n) U_j^\dg(n)\}_{n\in \bbz_+}$ and $\{V_i(n) V_j^\dg(n)\}_{n\in \bbz_+}$ become Cauchy sequences in the complete spaces $\Un$ and $\Um$, respectively. Thus, the limits of $U_iU_j^\dg$ and $V_iV_j^\dg$ exist. In particular, if we denote 
\begin{equation*}
X_i^\infty := \lim_{t\to\infty} U_i (t) U_1^\dg(t) ,
\end{equation*}
then one has
\[ \lim_{t\to\infty} U_i(t) U_j^\dg (t) =\lim_{t\to\infty}  U_iU_1^\dg(t) (U_j(t) U_1^\dg(t))^\dg  = X_i^\infty X_j^{\infty,\dg}, \quad \lim_{t\to\infty} \langle U_i,U_j\rangle_\tF = \langle X_i^\infty, X_j^\infty \rangle_\tF.
\]
On the other hand, we recall \eqref{AA-2-1} 
\begin{align} \label{Z-85}
\begin{aligned}
&\frac{\d}{\dt} (U_iU_j^\dg) = \mi (U_i U_j^\dg H_j - H_i U_iU_j^\dg ) \\
&+ \frac\kp N \sum_{k=1}^N \Big( \langle V_i,V_k\rangle_\tF U_kU_j^\dg    - \langle V_k,V_i\rangle_\tF U_iU_k^\dg U_i U_j^\dg   +  \langle V_k,V_j\rangle_\tF U_iU_k^\dg  - \langle V_j,V_k\rangle_\tF U_iU_j^\dg U_k U_j^\dg\Big).
\end{aligned}
\end{align}
In \eqref{Z-85}, if we let $t\to\infty$ and apply Barbalat's lemma, the left-hand side of \eqref{Z-85} vanishes and consequently, the relation \eqref{Z-85} becomes
\begin{align} \label{Z-86}
\begin{aligned}
O &= \mi (X_i^\infty X_j^{\infty,\dg} H_j - H_i X_i^\infty X_j^{\infty,\dg} ) \\
&\hspace{0.5cm}+ \frac\kp N \sum_{k=1}^N \Big(    \langle Y_i^\infty, Y_k^\infty\rangle _ \tF X_k^\infty X_j^{\infty,\dg} - \langle Y_k^\infty , Y_i^{\infty, \dg} \rangle_\tF X_i^\infty X_k^{\infty,\dg}X_i^\infty X_j^{\infty,\dg}  \\
&\hspace{0.5cm}+  \langle Y_k^\infty, Y_j^\infty\rangle _ \tF X_i^\infty X_k^{\infty,\dg} - \langle Y_j^\infty , Y_k^{\infty, \dg}\rangle_\tF X_i^\infty X_j^{\infty,\dg}X_k^\infty X_j^{\infty,\dg} \Big) .
\end{aligned} 
\end{align}
By performing left-multiplication  $X_i^{\infty,\dg}$ and right-multiplication $X_j^\infty$ with \eqref{Z-86}, we obtain 
\begin{align} \label{Z-87}
\begin{aligned}
& -\mi X_j^{\infty, \dg} H_j X_j^\infty + \frac\kp N \sum_{k=1}^N  \Big( \langle Y_j^\infty, Y_k^\infty\rangle _ \tF X_j^{\infty,\dg} X_k^\infty - \langle Y_k^\infty, Y_j^\infty\rangle _ \tF X_k^{\infty,\dg} X_j^\infty \Big) \\
 &\hspace{0.5cm} =    -\mi X_i^{\infty, \dg} H_i X_i^\infty + \frac\kp N \sum_{k=1}^N \Big( \langle Y_i^\infty, Y_k^\infty\rangle _ \tF X_i^{\infty,\dg} X_k^\infty - \langle Y_k^\infty, Y_i^\infty\rangle _ \tF X_k^{\infty,\dg} X_i^\infty\Big).
\end{aligned}
\end{align}
Since the relation\eqref{Z-87} does not depend on the index, we can set
\begin{equation*}
-\mi \Theta :=  -\mi X_j^{\infty, \dg} H_j X_j^\infty + \frac\kp N \sum_{k=1}^N  \Big( \langle Y_j^\infty, Y_k^\infty\rangle _ \tF X_j^{\infty,\dg} X_k^\infty - \langle Y_k^\infty, Y_j^\infty\rangle _ \tF X_k^{\infty,\dg} X_j^\infty \Big).
\end{equation*}
This yields
\begin{equation*}
X_j^\infty \Theta X_j^{\infty,\dg} = H_j + \frac\kp N \sum_{k=1}^N \Big( \langle Y_k^\infty, Y_j^\infty\rangle _ \tF X_k^{\infty,\dg} X_j^\infty - \langle Y_j^\infty, Y_k^\infty\rangle _ \tF X_j^{\infty,\dg} X_k^\infty \Big).
\end{equation*}
Similarly for $\{V_i\}$, we can define  a matrix $\Gamma$ independent of the index such that 
\begin{equation*}
Y_j^\infty \Gamma  Y_j^{\infty,\dg} =G_j + + \frac\kp N \sum_{k=1}^N  \Big( \langle X_k^\infty, X_j^\infty\rangle _ \tF Y_k^{\infty,\dg} Y_j^\infty - \langle X_j^\infty, X_k^\infty\rangle _ \tF Y_j^{\infty,\dg} Y_k^\infty \Big).
\end{equation*}
Then, $ \{ \{X_i^\infty\}, \Theta\}$ and $ \{ \{Y_i^\infty\}, \Gamma\}$ indeed consist of the phase-locked state. Moreover, one has
\[
\mathcal D(\mathcal X^\infty) = \max_{1\leq i,j\leq N } \|X_i^\infty X_j^{\infty,\dg} - I_n\| = \lim_{t\to\infty} \max_{1\leq i,j\leq N } \|U_i(t) U_j^\dg(t) -I_n\|_\tF = \lim_{t\to\infty} \mathcal D(\mathcal U) <\nu_2.
\]
Exactly the same estimate holds for $\mathcal D(\mathcal Y^\infty)$ as well.
\end{proof}

\begin{remark}
In \cite{Lo09}, the phase-locked state of \eqref{Z-0} are defined to be of the following form:
\begin{equation*}
U_i(t) = U_i^\infty e^{-\mi \Gamma_U t}, \quad V_i(t) = V_i^\infty e^{-\mi \Gamma_V t},
\end{equation*} 
where $U_i^\infty \in \Un$ and $V_i^\infty \in \Um$  are unitary matrices, and $\Lambda_U$ and $\Lambda_V$ satisfy 
\begin{align*} \label{pl}
\begin{aligned}
&U_i^\infty \Gamma_U U_i^{\infty,\dg} = H_i + \frac{\mi \kp}{N} \sum_{k=1}^N\Big( \langle V_k^\infty, V_j^\infty\rangle_\tF  U_k^\infty U_i^{\infty,\dg} -\langle V_j^\infty, V_k^\infty\rangle_\tF U_i^\infty U_k^{\infty,\dg} \Big), \\
&V_i^\infty \Gamma_V V_i^{\infty,\dg} = G_i + \frac{\mi \kp}{N} \sum_{k=1}^N\Big( \langle U_k^\infty, U_j^\infty\rangle_\tF  V_k^\infty V_i^{\infty,\dg} -\langle U_j^\infty, U_k^\infty\rangle_\tF V_i^\infty V_k^{\infty,\dg} \Big) . 
\end{aligned}
\end{align*}
\end{remark}
\section{Rank-$2m$ quadratically separable state} \label{sec:5} 
\setcounter{equation}{0}
In this section, we study a quadratically separable state of the LT model for rank-$2m$ tensors by introducing the  MM model whose solution configuration is given as follows.
\begin{equation*}
\{ \mathcal U^1,\mathcal U^2,\cdots, \mathcal U^m\},\quad \mathcal U^p := (U_1^p,\cdots,U_N^p),\quad p=1,\cdots,m.
\end{equation*}
In the following two subsections, we introduce extended models for the DM model \eqref{C-10-1} and the DUM model \eqref{C-17}. Since the procedures are similar as those in Section \ref{sec:3}, we omit details. 
\subsection{The  MM model} \label{sec:5.1}
In this subsection, we introduce the MM model:
\begin{align}\label{C-15-1}
\begin{cases}
 \displaystyle\dot{U}_j^p=B_j^p U_j^p + \displaystyle\frac{\kappa_1}{N}\sum_{k=1}^N\left(\prod_{\substack{\ell=1\\\ell\neq p}}^m\langle U_j^\ell, U_k^\ell\rangle_\tF  U_k^p\big(U_j^p\big)^\dagger U_j^p-\prod_{\substack{\ell=1\\\ell\neq p}}^m\langle U_k^l, U_j^l\rangle_\tF U_j^p\big(U_k^p\big)^\dagger U_j^p\right)\\
 \hspace{1cm}\displaystyle+\frac{\kappa_2}{N}\sum_{k=1}^N\left(\prod_{\substack{\ell=1\\\ell\neq p}}^m\langle U_j^\ell, U_k^\ell\rangle_\tF  U_j^p\big(U_j^p\big)^\dagger U_k^p-\prod_{\substack{\ell=1\\\ell\neq p}}^m\langle U_k^\ell, U_j^\ell\rangle_\tF U_j^p\big(U_k^p\big)^\dagger U_j^p\right),
\end{cases}
\end{align}
where $B_j^p$ is a rank-$4m$ tensor satisfying the skew-symmetric property. \newline

\noindent For a solution $\{ \mathcal U^p\}_{p=1}^m$ to   \eqref{C-15-1}, we set a rank-$2m$ tensor denoted by $T_i$:
\begin{equation} \label{C-15-2} 
T_i(t)=U_i^1(t)\otimes U_i^2(t)\otimes \cdots\otimes U_i^m(t),\quad U_i^p\in\bbc^{d_1^p\times d_2^p}, \quad i=1,\cdots,N,\quad p=1,\cdots,m,
\end{equation}
which can be written in a component form:
\begin{equation*}
[T_i]_{\alpha_1\beta_1\alpha_2\beta_2\cdots\alpha_m\beta_m}=[U_i^1]_{\alpha_1\beta_1}[U_i^2]_{\alpha_2\beta_2}\cdots[U_i^m]_{\alpha_m\beta_m}.
\end{equation*}
In order to relate the model  \eqref{C-15-1} with the LT model \eqref{LT}, we consider the index vector $i_*$ in $\kp_{i_*}$. Then, we introduce two subsets of $\{0,1\}^{2d}$: for $q =1,\cdots,m$, 
\begin{align*}
\Lambda_1&=\{i_*\in\{0, 1\}^{2d}: \text{0 appears once at $2q-1$-th coordinate}\},\\
\Lambda_2&=\{i_*\in\{0, 1\}^{2d}: \text{0 appears once at $2q$-th coordinate}\}.
\end{align*}
Note that $|\Lambda_1|= |\Lambda_2| = m$ and for instance with $m=2$, 
\begin{equation*}
(0,1,1,1), (1,1,0,1) \in \Lambda_1,\quad (1,0,1,1),(1,1,1,0) \in \Lambda_2.  
\end{equation*}
In this regard, we choose the index vector as 
\begin{align}\label{C-18}
\kappa_{i_*}=\begin{cases}
\kappa_1\quad&\text{if }i_*\in\Lambda_1,\\
\kappa_2\quad&\text{if }i_*\in\Lambda_2,\\
0&\text{otherwise},
\end{cases}
\end{align}
which generalizes \eqref{C-9-1} for the DM model. \newline

Next, for the natural frequency tensors, we use $\{B_j^p\}_{p=1}^m$ to associate a rank-$4m$ tensor $A_j$ as 
\begin{align}\label{C-18-1}
[A_j]_{\alpha_1\alpha_2\cdots\alpha_{2m}\beta_1\beta_2\cdots\beta_{2m}}=\sum_{k=1}^m\left([B_j^k]_{\alpha_{2k-1}\alpha_{2k}\beta_{2k-1}\beta_{2k}}\prod^d_{\substack{\ell=1\\ \ell\neq k }}\left(\delta_{\alpha_{2\ell-1}\beta_{2\ell-1}}\delta_{\alpha_{2\ell}\beta_{2\ell}}\right)\right),
\end{align}
which corresponds to \eqref{C-11} for the DM model. Then, it follows from straightforward calculation that the Lohe tensor model with \eqref{C-18} and \eqref{C-18-1} reduces to \eqref{C-15-1} whose solutions can be related by \eqref{C-15-2}. The argument above is summarized in the following proposition analogous to Proposition \ref{P3.1}.

\begin{proposition}\label{P5.1}
The following assertions hold. 
\begin{enumerate}
\item Suppose $\{\mathcal U^p\}_{p=1}^m$ is a solution to \eqref{C-15-1}. Then, a rank-$2m$ tensor defined by $T_i := U_i^1 \otimes U_i^2 \otimes \cdots \otimes U_i^m$ is the QSS to \eqref{LT} with  well-prepared initial data and free flow tensors $A_i$ satisfying \eqref{C-18-1}. 
    
\vspace{0.2cm}

\item
Suppose a rank-$2m$ tensor $T_i$ is a solution to \eqref{LT} with \eqref{C-18-1} and quadratically separable initial data:
\begin{equation*}
T_i^0  := U_i^{1,0} \otimes U_i^{2,0} \otimes \cdots \otimes U_i^{p,0},\quad 1\leq i \leq N ,
\end{equation*}
for rank-2 tensors $U_i^{p,0} \times \bbc^{d_1^p \times d_2^p}$ with unit norms. Then, there exist matrices $\{U_i^p\}$ with unit norms such that
\begin{equation*}
T_i(t) = U_i^1(t) \otimes U_i^2(t) \otimes \cdots \otimes U_i^p(t),\quad t>0,
\end{equation*}
where $\{U_i^p\}$ is a solution to \eqref{C-15-1} with $U_i^{p,0} = U_i^p(0)$. 
\end{enumerate}
\end{proposition}
\subsection{The  MUM model}  \label{sec:5.2}
In this subsection, we further reduce to the MM model to the model on the product of the unitary groups. Since the unitary group is concerned, we set 
\begin{equation} \label{C-20}
d_1^p=d_2^p=:d_p,\quad  1\leq p\leq m.
\end{equation}
For the modeling of natural frequencies, we also define Hermitian matrices with a size $d_p \times d_p$: 
\begin{equation*}\label{C-21}
[-\mathrm{i}H_j^p]_{\alpha_1\alpha_2}\delta_{\beta_1\beta_2} := [B_j^p]_{\alpha_1\beta_1\alpha_2\beta_2}.
\end{equation*}
In addition, $\{H_j^p\}$ satisfy
\begin{align}\label{C-24}
[A_j]_{\alpha_1\alpha_2\cdots\alpha_{2m}\beta_1\beta_2\cdots\beta_{2m}}
&=\sum_{k=1}^m\left([-\mathrm{i}H_j^k]_{\alpha_{2k-1}\beta_{2k-1}}\delta_{\alpha_{2k}\beta_{2k}}\prod^m_{\substack{\ell=1\\\ell\neq k }}\left(\delta_{\alpha_{2\ell-1}\beta_{2\ell-1}}\delta_{\alpha_{2\ell}\beta_{2\ell}}\right)\right).
\end{align}
By Lemma \ref{L3.2}, one can verify that system \eqref{C-15-1} with \eqref{C-20} conserves the unitarity of $U_j^p$. Thus, system \eqref{C-15-1} reduces to the following model on the unitary group:
\begin{align}\label{C-23}
\begin{cases}
\displaystyle\dot{U}_j^p=-\mi H_j^pU_j^p+\displaystyle\frac{\kappa}{N}\sum_{k=1}^N\left(\prod_{\substack{\ell=1\\\ell\neq p}}^m\langle U_j^l, U_k^l\rangle_\tF  U_k^p -\prod_{\substack{\ell=1\\\ell\neq p}}^m\langle U_k^\ell, U_j^\ell\rangle_\tF U_j^p\big(U_k^p\big)^\dagger U_j^p\right),\\
U_j^p(0)=U_j^{p,0}\in\mathbf{U}(d_p),
\end{cases}
\end{align}
where $\kappa=\kappa_1+\kappa_2$ and $H_j^p U_j^p$ is a usual matrix product. As in Proposition \ref{P5.1},  existence and uniqueness of the QSS for \eqref{C-23} can be stated as follows. 
\begin{proposition} \label{P5.2}
The following assertions hold. 
\begin{enumerate}
\item Suppose $\{\mathcal U^p\}_{p=1}^m$ is a solution to \eqref{C-23}. Then, a rank-$2m$ tensor defined by $T_i := U_i^1 \otimes U_i^2 \otimes \cdots \otimes U_i^m$ is a quadratically separable state to \eqref{LT} with a well-prepared free flow tensor $A_i$ satisfying \eqref{C-24}. 

\vspace{0.2cm}

\item
Suppose a rank-$2m$ tensor $T_i$ is a solution to \eqref{LT} with \eqref{C-24} and quadratically separable initial data:
\begin{equation*}
T_i^0  := U_i^{1,0} \otimes U_i^{2,0} \otimes \cdots \otimes U_i^{p,0},\quad 1\leq i \leq N ,
\end{equation*}
for rank-2 tensors $U_i^{p,0} \times \bbc^{n_1^p \times n_2^p}$ with unit norms. Then, there exist matrices $\{U_i^p\}$ with unit norms such that
\begin{equation*}
T_i(t) = U_i^1(t) \otimes U_i^2(t) \otimes \cdots \otimes U_i^p(t),\quad t>0,
\end{equation*}
where $\{U_i^p\}$ is a solution to \eqref{C-23} with $U_i^{p,0} = U_i^p(0)$. 
\end{enumerate}
\end{proposition}
\begin{remark}
In Section \ref{sec:4}, we provided the emergent dynamics of the double unitary matrix model \eqref{Z-0}. However for its generalized model \eqref{C-23}, emergent dynamics will not be studied, since it can be straightforwardly obtained from the results in Section \ref{sec:4}. 
\end{remark}

\section{Conclusion} \label{sec:6}
\setcounter{equation}{0}
In this paper, we have studied the existence and emergent dynamics of quadratically separable states for the Lohe tensor model which incorporates several well-known low-rank aggregation models such as the Kuramoto model, the Lohe sphere model and the Lohe matrix model, etc. In our previous work \cite{H-K-P}, we obtained completely separable states as special solutions to the Lohe tensor model defined as  tensor products of rank-1 real tensors (or vectors). In analogy with the previous work, we consider the states in which a solution can be decomposed as a tensor product of rank-2   tensors (or matrices), namely, a quadratically separable state. Precisely, if   initial data are quadratically separable, then such separability is preserved along the Lohe tensor flow. Moreover, by introducing and analyzing double matrix and unitary matrix models, we are able to study asymptotic behavior of the quadratically separable states to the Lohe tensor model. Of course, there are several issues that are not discussed in this work. For example, one can naturally consider the state consisting of tensors with possibly different ranks and sizes. We explore this issue in a future work. 

\newpage

\appendix
\section{Proof of Lemma \ref{L4.1}} \label{sec:app.A}
\setcounter{equation}{0}
In this appendix, we provide a proof of Lemma \ref{L4.1} in which a differential inequality for the aggregation functional $\mathcal L=\mathcal D(\mathcal U) + \mathcal D(\mathcal V) + \mathcal S(\mathcal U) +\mathcal  S(\mathcal V)$ is derived. We divide a proof into two steps: \newline
\begin{itemize}
\item
Step A:~we derive differential inequalities for $\mathcal D(\mathcal U)$ and $ \mathcal D(\mathcal V)$ (see Lemma \ref{LA.1}).

\vspace{0.2cm}

\item
Step B:~ we derive differential inequalities for $\mathcal S(\mathcal U)$ and $\mathcal S(\mathcal V)$ (see Lemma  \ref{LA.2}).
\end{itemize}

\vspace{0.5cm}

\begin{lemma} \label{LA.1} 
Let $\{(U_i,V_i)\}$ be a solution to \eqref{Z-1}. Then, $\mathcal D(\mathcal U)$ and $\mathcal D(\mathcal V) $ satisfy 
\begin{align} \label{AA-2} 
\begin{aligned}
&\frac{\d}{\dt} \mathcal D(\mathcal U) \leq -2m \kp \mathcal D(\mathcal U) + m\kp \mathcal D(\mathcal U)^3 + 6\kp \mathcal S(\mathcal V) \mathcal D(\mathcal U) + 2\kp \mathcal S(\mathcal V) \mathcal D(\mathcal U)^2 + 4\kp \sqrt n S(\mathcal V) , \\
&\frac{\d}{\dt} \mathcal D(\mathcal V) \leq -2n \kp \mathcal D(\mathcal V) + n\kp \mathcal D(\mathcal V)^3 + 6\kp \mathcal S(\mathcal U) \mathcal D(\mathcal V) + 2\kp \mathcal S(\mathcal U) \mathcal D(\mathcal V)^2 + 4\kp \sqrt m S(\mathcal U) .
\end{aligned}
\end{align}
\end{lemma} 
\begin{proof} 
 By straightforward calculations, one finds a differential equation for $G_{ij} = U_iU_j^\dg$:
\begin{equation} \label{AA-2-1}
\frac{\d}{\dt} G_{ij} = \frac\kp N \sum_{k=1}^N (c_{ik} G_{kj} - c_{ki} G_{ik} G_{ij} + c_{kj} G_{ik} - c_{jk} G_{ij} G_{kj}).
\end{equation}
By using the relation $G_{ij} = I_n - S_{ij}$, we see that $S_{ij}$ satisfies
\begin{align} \label{AA-3}
\begin{aligned}
\frac{\d}{\dt} S_{ij} &=  \frac\kp N \sum_{k=1}^N \Big[  c_{ik} S_{kj} - c_{ki}S_{ik} - c_{ki}S_{ij} + c_{i}S_{k}S_{ij} + c_{kj}S_{ik} - c_{jk}S_{kj} - c_{jk}S_{ij} + c_{jk} S_{ij}S_{kj} \\
&\hspace{2.5cm} + (c_{ki} - c_{ik} + c_{jk} - c_{kj})I_n \Big ].
\end{aligned}
\end{align}
After algebraic manipulation, we rewrite \eqref{AA-3} in terms of $S_{ij}$ and $c_{ij}-m$ that are expected to converge to zero:
\begin{align} \label{AA-4}
\begin{aligned}
\frac{\d S_{ij}}{dt} & =-2 m \kp S_{ij} + \frac{m\kp}{N} \sum_{k=1}^N \Big( S_{ik} S_{ij} + S_{ij}S_{kj}  \Big)\\
& \hspace{0.2cm}+ \frac\kp N \sum_{k=1}^N \Big[ (c_{ik} - m)S_{kj} - (c_{ki}-m)S_{ik} - (c_{ik} - m )S_{ij} + (c_{ki}-m)S_{ik}S_{ij} \Big ] \\
&\hspace{0.2cm}+\frac\kp N \sum_{k=1}^N \Big[ (c_{kj} - m )S_{ik} - (c_{jk}-m)S_{kj} - (c_{jk}-m)S_{ij} + (c_{jk}-m)S_{ij}S_{kj} \Big ] \\
&\hspace{0.2cm}+\frac\kp N \sum_{k=1}^N \Big[  (c_{ki} - c_{ik}  +  c_{jk} - c_{kj})I_n \Big ].
\end{aligned}
\end{align}
On the other hand for an $n\times n$ matrix $A$, one has
\begin{equation*}
\frac12\frac{\d}{\dt} \|A\|_\tF^2 = \frac12\frac{\d}{\dt} \textup{tr} (AA^\dg)  = \frac12\tr (\dot A A^\dg + A\dot A^\dg) = \Re\tr (\dot A A^\dg).
\end{equation*}
We multiply \eqref{AA-4} with $S_{ij}^\dg$ to find
\begin{align} \label{AA-6}
\begin{aligned}
&\frac12\frac{\d}{\dt} \|S_{ij}\|_\tF^2 =  -2m\kp \|S_{ij}\|_\tF^2 + \frac{m\kp}{N}\sum_{k=1}^N \Re\tr ( S_{ik} S_{ij} S_{ij}^\dg + S_{ij} S_{kj} S_{ij}^\dg) \\ 
&\hspace{0.5cm} +\frac\kp N \sum_{k=1}^N \Re [(c_{ik}-m) \tr(  S_{kj}S_{ij}^\dg )] - \Re[(c_{ki}-m)\tr ( S_{ik}S_{ij}^\dg)] - \Re(c_{ik}-m)\|S_{ij}\|_\tF^2 \\
&\hspace{0.5cm}+ \frac\kp N \sum_{k=1}^N \Re[(c_{kj}-m) \tr(S_{ik}S_{ij}^\dg)] - \Re[(c_{jk}-m)\tr(S_{kj}S_{ij}^\dg) ]- \Re(c_{jk}-m) \|S_{ij}\|_\tF^2 \\
&\hspace{0.5cm}+ \frac\kp N \sum_{k=1}^N \Re[(c_{ki}-m) \tr( S_{ik}S_{ij}S_{ij}^\dg)] + \Re[ (c_{jk}-m)\tr (S_{ij}S_{kj}S_{ij}^\dg) ]\\
&\hspace{0.5cm}+\frac\kp N \sum_{k=1}^N  \Re[(c_{ki}-c_{ik} + c_{jk} - c_{kj}) \tr(S_{ij}^\dg) ]\\
& = : -2m\kp \|S_{ij}\|_\tF^2 + \mathcal I_{11} + \mathcal I_{12} +  \mathcal I_{13} +  \mathcal I_{14} +  \mathcal I_{15}.
\end{aligned}
\end{align}
Below, we present   estimates for $\mathcal I_{1k},~k=1,\cdots,5$, respectively. \newline

\noindent $\bullet$ (Estimate of $\mathcal I_{11}$): We use \eqref{AA-7} and 
\begin{equation} \label{AA-7}
S_{ij} + S_{ji} = S_{ij}S_{ji}, \quad S_{ij}^\dg = S_{ji} 
\end{equation} 
to derive 
\begin{align*}
\Re \tr ( S_{ik} S_{ij} S_{ji} ) &= \frac12 \tr ( S_{ik} S_{ij} S_{ji} + S_{ij}S_{ji}S_{ki}) = \frac12 \tr( S_{ij}S_{ji}(S_{ki} + S_{ik}) )  \\
& = \frac12 \tr (S_{ij}S_{ji} S_{ki}S_{ik}) = \frac12\|S_{ki}S_{ij}\|_\tF^2. 
\end{align*}
Similarly,  one has 
\begin{equation*}
\Re\tr ( S_{ij}S_{kj} S_{ij}^\dg) = \frac12 \|S_{ij}S_{jk}\|_\tF^2 .
\end{equation*}
Hence, $\mathcal I_{11}$ satisfies
\[
\mathcal I_{11} = \frac{m\kp}{N}\sum_{k=1}^N \Re\tr ( S_{ik} S_{ij} S_{ij}^\dg + S_{ij} S_{kj} S_{ij}^\dg) = \frac{m\kp}{2N}\sum_{k=1}^N  \big( \|S_{ki}S_{ij}\|_\tF^2 + \|S_{ij}S_{jk}\|_\tF^2 \big) \leq m \kp \mathcal D(\mathcal U)^4.
\]
$\bullet$ (Estimates of $\mathcal I_{12}$ and $\mathcal I_{13}$): By the maximality of $\mathcal D(\mathcal U)$ and $\mathcal S(\mathcal V)$, we have
\begin{align*}
\mathcal I_{12} &= \frac\kp N \sum_{k=1}^N\Re[ (c_{ik}-m) \tr(  S_{kj}S_{ij}^\dg ) ]- \Re[(c_{ki}-m)\tr ( S_{ik}S_{ij}^\dg) ]- \Re(c_{ik}-m)\|S_{ij}\|_\tF^2 \\
&\leq 3\kp \mathcal S(\mathcal V) \mathcal D(\mathcal U)^2, \\
\mathcal I_{13} & =  \frac\kp N \sum_{k=1}^N\Re[ (c_{kj}-m) \tr(S_{ik}S_{ij}^\dg)] -\Re[ (c_{jk}-m)\tr(S_{kj}S_{ij}^\dg) ]-\Re (c_{jk}-m) \|S_{ij}\|_\tF^2 \\
& \leq 3\kp\mathcal S(\mathcal V) \mathcal D(\mathcal U)^2.  
\end{align*}
$\bullet$ (Estimate of $\mathcal I_{14}$): By straightforward calculation, one has
\begin{align*}
\mathcal I_{14} &=  \frac\kp N \sum_{k=1}^N \Re[(c_{ki}-m) \tr( S_{ik}S_{ij}S_{ij}^\dg)] + \Re[ (c_{jk}-m)\tr (S_{ij}S_{kj}S_{ij}^\dg) ]  \leq  2\kp \mathcal S(\mathcal V) \mathcal D(\mathcal U)^3.
\end{align*} 
$\bullet$ (Estimate of $\mathcal I_{15}$): Note that for an $n\times n$ matrix $A$, 
\begin{equation*}
\tr(A) = \tr(I_n A) \leq \|I_n\|_\tF \|A\|_\tF = \sqrt n \|A\|_\tF.
\end{equation*}
This yields
\begin{equation} \label{AA-9}
\mathcal I_{15} = \frac\kp N \sum_{k=1}^N  \Re[(c_{ki}-c_{ik} + c_{jk} - c_{kj}) \tr(S_{ij}^\dg) ] \leq 4 \kp\sqrt n  \mathcal S(\mathcal V) \mathcal D(\mathcal U).
\end{equation}
In \eqref{AA-6}, we collect all the estimates for $\mathcal I_{1k},~k=1,\cdots,5$ to obtain 
\[
\frac12\frac{\d}{\dt} \|S_{ij}\|_\tF^2 \leq   -2m\kp \|S_{ij}\|_\tF^2 + m\kp \mathcal D(\mathcal U)^4 + 6\kp \mathcal S(\mathcal V) \mathcal D(\mathcal U)^2 + 2\kp \mathcal S(\mathcal V) \mathcal D(\mathcal U)^3 + 4\kp \sqrt n \mathcal S(\mathcal V) \mathcal D(\mathcal U).
\]
Hence, $\mathcal D(\mathcal U)$ satisfies
\begin{equation*} \label{AA-10}
\frac{\d}{\dt} \mathcal D(\mathcal U) \leq -2m \kp \mathcal D(\mathcal U) + m\kp \mathcal D(\mathcal U)^3 + 6\kp \mathcal S(\mathcal V) \mathcal D(\mathcal U) + 2\kp \mathcal S(\mathcal V) \mathcal D(\mathcal U)^2 + 4\kp \sqrt n \mathcal S(\mathcal V) .
\end{equation*}
Similarly, one can find a differential inequality for $\mathcal D(\mathcal V)$ by exchanging the roles of $\mathcal U$ and $\mathcal V$: 
\begin{equation*} \label{AA-11} 
\frac{\d}{\dt} \mathcal D(\mathcal V) \leq -2n \kp \mathcal D(\mathcal V) + n\kp \mathcal D(\mathcal V)^3 + 6\kp \mathcal S(\mathcal U) \mathcal D(\mathcal V) + 2\kp \mathcal S(\mathcal U) \mathcal D(\mathcal V)^2 + 4\kp \sqrt m\mathcal S(\mathcal U) .
\end{equation*}
\end{proof}
In \eqref{AA-2}, note that $\mathcal S(\mathcal U)$ and $\mathcal S(\mathcal V)  $ appear in the differential inequalities for $\mathcal D(\mathcal U)$ and $\mathcal D(\mathcal V) $. Hence, we derive the differential inequalities for $\mathcal S(\mathcal U)$ and $\mathcal S(\mathcal V) $ below. 

\begin{lemma} \label{LA.2} 
Let $\{(U_i,V_i)\}$ be a solution to \eqref{Z-1}. Then, $\mathcal S(\mathcal U)$ and $\mathcal S(\mathcal V) $ satisfy 
\begin{align} \label{AA-12}
\begin{aligned}
&\frac{\d}{\dt} \mathcal S(\mathcal U) \leq -2m\kp \mathcal S(\mathcal U) + 2m\kp \mathcal D( \mathcal U)^2 + 6\kp \mathcal S(\mathcal U) \mathcal S(\mathcal V) + 2\kp \mathcal S(\mathcal V) \mathcal D(\mathcal U)^2 + 4\kp \sqrt n \mathcal S(\mathcal V), \\
&\frac{\d}{\dt} \mathcal S(\mathcal V) \leq -2n\kp \mathcal S(\mathcal V) + 2n\kp \mathcal D( \mathcal V)^2 + 6\kp \mathcal S(\mathcal U) \mathcal S(\mathcal V) + 2\kp \mathcal S(\mathcal U) \mathcal D(\mathcal V)^2 + 4\kp \sqrt m \mathcal S(\mathcal U).
\end{aligned}
\end{align}
\end{lemma}
\begin{proof} 
 In \eqref{AA-4}, we take trace to obtain 
\begin{align} \label{AA-15}
\begin{aligned}
&\frac{\d}{\dt} (n- d_{ij}) =-2m\kp ( n-d_{ij}) + \frac{m\kp}{N}\sum_{k=1}^N \textup{tr} (S_{ik}S_{ij} + S_{ij} S_{kj}) \\
&\hspace{1cm}+ \frac\kp N \sum_{k=1}^N (c_{ik} - m)(n-d_{kj}) - (c_{ki}-m)(n-d_{ik}) - (c_{ik}-m)(n-d_{ij}) \\
&\hspace{1cm}+\frac\kp N \sum_{k=1}^N (c_{kj}-m)(n-d_{ik}) - (c_{jk}-m)(n-d_{kj}) - (c_{jk}-m)(n-d_{ij}) \\
&\hspace{1cm} +\frac\kp N \sum_{k=1}^N (c_{ki}-m)\textup{tr} (S_{ik}S_{ij}) + (c_{jk}-m)\textup{tr}( S_{ij}S_{kj}) \\
&\hspace{1cm} + \frac\kp N \sum_{k=1}^N[  (c_{ki} - c_{ik}) + (c_{jk} - c_{kj})]\sqrt n  \\
&\hspace{1cm} =: -2m\kp ( n-d_{ij}) + \mathcal I_{21} +  \mathcal I_{22}+  \mathcal I_{23} +  \mathcal I_{24} +  \mathcal I_{25}. 
\end{aligned}
\end{align}
Below, we present   estimates of $\mathcal I_{2k},~k=1,\cdots,5$, separately. \newline

\noindent $\bullet$ (Estimate of $\mathcal I_{21}$): We use the definition of $\mathcal D(\mathcal U)$ to find
\begin{align*}
\mathcal I_{21} = \frac{m\kp}{N}\sum_{k=1}^N \textup{tr} (S_{ik}S_{ij} + S_{ij} S_{kj}) \leq 2m\kp \mathcal D(\mathcal U)^2.
\end{align*}
\noindent $\bullet$ (Estimates of $\mathcal I_{22}$ and $\mathcal I_{23}$):~It is easy to see that  
\begin{align*}
&\mathcal I_{22} = \frac\kp N \sum_{k=1}^N (c_{ik} - m)(n-d_{kj}) - (c_{ki}-m)(n-d_{ik}) - (c_{ik}-m)(n-d_{ij}) \leq 3\kp \mathcal S(\mathcal U) \mathcal S(\mathcal V), \\
&\mathcal I_{23} = \frac\kp N \sum_{k=1}^N (c_{kj}-m)(n-d_{ik}) - (c_{jk}-m)(n-d_{kj}) - (c_{jk}-m)(n-d_{ij}) \leq 3\kp \mathcal S(\mathcal U) \mathcal S(\mathcal V).
\end{align*}
\noindent $\bullet$ (Estimate of $\mathcal I_{24}$): Similar to $\mathcal I_{21}$, one has
\begin{align*}
\mathcal I_{24} = \frac\kp N \sum_{k=1}^N (c_{ki}-m)\textup{tr} (S_{ik}S_{ij}) + (c_{jk}-m)\textup{tr}( S_{ij}S_{kj}) \leq 2\kp \mathcal S(\mathcal V) \mathcal D(\mathcal U)^2. 
\end{align*} 
\noindent $\bullet$ (Estimate of $\mathcal I_{25}$): We find
\begin{equation} \label{AA-17}
\mathcal I_{25} = \frac\kp N \sum_{k=1}^N[  (c_{ki} - c_{ik}) + (c_{jk} - c_{kj})]\sqrt n \leq 4\kp \sqrt n \mathcal S(\mathcal V).
\end{equation} 
In \eqref{AA-15}, we combine all estimates to obtain 
\begin{equation*} \label{AA-20}
\frac{\d}{\dt} \mathcal S(\mathcal U) \leq  -2m\kp \mathcal S(\mathcal U) + 2m\kp \mathcal D( \mathcal U)^2 + 6\kp \mathcal S(\mathcal U) \mathcal S(\mathcal V) + 2\kp \mathcal S(\mathcal V) \mathcal D(\mathcal U)^2 + 4\kp \sqrt n \mathcal S(\mathcal V).
\end{equation*}
By exchanging the roles of $\mathcal U$ and $\mathcal V$, we derive a differential inequality for $\mathcal S(\mathcal V)$: 
\begin{equation*} \label{AA-21}
\frac{\d}{\dt} \mathcal S(\mathcal V) \leq -2n\kp \mathcal S(\mathcal V) + 2n\kp \mathcal D( \mathcal V)^2 + 6\kp \mathcal S(\mathcal U) \mathcal S(\mathcal V) + 2\kp \mathcal S(\mathcal U) \mathcal D(\mathcal V)^2 + 4\kp \sqrt m \mathcal S(\mathcal U).
\end{equation*}
\end{proof}

\begin{remark} \label{RA.1} 
For homogeneous Hamiltonians, we add $\eqref{AA-12}_1$ and $\eqref{AA-12}_2$ to find 
\begin{align*}
\frac{\d}{\dt} (\mathcal S(\mathcal U) + \mathcal S(\mathcal V)) &\leq  -2\kp ( m-2\sqrt m ) \mathcal S(\mathcal U) -2\kp ( n-2\sqrt n) \mathcal S(\mathcal V) + \mathcal O( (\mathcal S(\mathcal U) + \mathcal S(\mathcal V))^2).
\end{align*}
Hence, in order to obtain the desired convergence, we assume
\begin{equation*}
n>2\sqrt n,\quad m >2\sqrt m,\quad \textup{i.e.,}\quad n,m>4,
\end{equation*}
 which requires the restriction on the size of $U_i$ and $V_j$. This technical assumption arises from the estimate of $\mathcal I_{25}$ in \eqref{AA-17}.
\end{remark}
\vspace{0.5cm}

\noindent Now, we are ready to present a proof of Lemma \ref{L4.1} using  Lemma \ref{LA.1} and Lemma \ref{LA.2}. \newline

\noindent {\it Proof of Lemma \ref{L4.1}}:~Recall the inequalities in \eqref{AA-2} and \eqref{AA-12}:
\begin{align} \label{AA-24}
\begin{aligned}
&\frac{\d}{\dt} \mathcal D(\mathcal U) \leq -2m \kp \mathcal D(\mathcal U) + m\kp \mathcal D(\mathcal U)^3 + 6\kp \mathcal S(\mathcal V) \mathcal D(\mathcal U) + 2\kp \mathcal S(\mathcal V) \mathcal D(\mathcal U)^2 + 4\kp \sqrt n \mathcal S(\mathcal V), \\
&\frac{\d}{\dt} \mathcal D(\mathcal V) \leq -2n \kp \mathcal D(\mathcal V) + n\kp \mathcal D(\mathcal V)^3 + 6\kp \mathcal S(\mathcal U) \mathcal D(\mathcal V) + 2\kp \mathcal S(\mathcal U) \mathcal D(\mathcal V)^2 + 4\kp \sqrt m \mathcal S(\mathcal U), \\
&\frac{\d}{\dt} \mathcal S(\mathcal U) \leq -2m\kp \mathcal S(\mathcal U) + 2m\kp \mathcal D( \mathcal U)^2 + 6\kp \mathcal S(\mathcal U) \mathcal S(\mathcal V) + 2\kp \mathcal S(\mathcal V) \mathcal D(\mathcal U)^2 + 4\kp \sqrt n \mathcal S(\mathcal V), \\
&\frac{\d}{\dt} \mathcal S(\mathcal V) \leq -2n\kp \mathcal S(\mathcal V) + 2n\kp \mathcal D( \mathcal V)^2 + 6\kp \mathcal S(\mathcal U) \mathcal S(\mathcal V) + 2\kp \mathcal S(\mathcal U) \mathcal D(\mathcal V)^2 + 4\kp \sqrt m \mathcal S(\mathcal U).
\end{aligned}
\end{align} 
Without loss of generality, we may assume $n\geq m$, and add all the inequalities in \eqref{AA-24} to find 
\begin{align} \label{AA-27}
\begin{aligned}
\frac{\d}{\dt} \mathcal L &\leq  - 2\kp (m-4\sqrt n) \mathcal L + n\kp (\mathcal D(\mathcal U)^3 + 2\mathcal D(\mathcal U)^2 + \mathcal D(\mathcal V)^3 + 2\mathcal D(\mathcal V)^2) \\
&\hspace{0.5cm} +6\kp ( \mathcal D(\mathcal U) \mathcal S(\mathcal V) + \mathcal D(\mathcal V)\mathcal S(\mathcal U) + 2 \mathcal S(\mathcal U)\mathcal S(\mathcal V)) + 4\kp ( \mathcal D(\mathcal U)^2 \mathcal S(\mathcal V) + \mathcal D(\mathcal V)^2 \mathcal S(\mathcal U)) \\
&=:- 2\kp (m-4\sqrt n) \mathcal L + n \kp \mathcal I_{31} + 6\kp \mathcal I_{32} + 4\kp \mathcal I_{33}. 
\end{aligned}
\end{align}
In the sequel, we provide estimates for $\mathcal I_{3k},~k=1,2,3$, respectively. \newline

\noindent $\bullet$ (Estimate of $\mathcal I_{31}$): We use a rough estimate to find
\begin{align*}
\mathcal I_{31} = \mathcal D(\mathcal U)^3 + 2\mathcal D(\mathcal U)^2 + \mathcal D(\mathcal V)^3 + 2\mathcal D(\mathcal V)^2 \leq \mathcal L^3 + 2\mathcal L^2 .
\end{align*}
\noindent $\bullet$ (Estimate of $\mathcal I_{32}$): By straightforward calculation, one has
\begin{align*}
\mathcal I_{32} &=  \mathcal D(\mathcal U) \mathcal S(\mathcal V) + \mathcal D(\mathcal V)\mathcal S(\mathcal U) + 2 \mathcal S(\mathcal U)\mathcal S(\mathcal V)  \\
&\leq \frac12 \mathcal D(\mathcal U)^2 +  \frac12\mathcal D(\mathcal V)^2 + \frac32 \mathcal S(\mathcal U)^2 +\frac32  \mathcal S(\mathcal V)^2    \leq \frac32\mathcal L^2.
\end{align*}
\noindent $\bullet$ (Estimate of $\mathcal I_{33}$):~We use Young's inequality that for $a,b>0$, 
\begin{equation*}
a^2 b \leq \frac23 a^3 + \frac13 b^3
\end{equation*}
to find
\begin{align*}
\mathcal I_{33} = \mathcal D(\mathcal U)^2 \mathcal S(\mathcal V) + \mathcal D(\mathcal V)^2 \mathcal S(\mathcal U) \leq \frac23 \mathcal D(\mathcal U)^3  + \frac13 \mathcal S(\mathcal V)^3 + \frac23 \mathcal D(\mathcal V)^3 + \frac13 \mathcal S(\mathcal U)^3 \leq \frac23 \mathcal L^3.
\end{align*}
In \eqref{AA-27}, we combine all the estimates to obtain the desired inequality for $\mathcal L$: 
\begin{align*}
\frac{\d}{\dt} \mathcal L &\leq-2\kp (m-4\sqrt n)\mathcal L + 2\kp n(\mathcal L^3 + 2\mathcal L^2) + 9\kp \mathcal L^2 + \frac{8\kp}{3} \mathcal L^3 \\
&= -2\kp (m-4\sqrt n)\mathcal L + + \kp(4n+9)\mathcal L^2 +  \kp \left( 2n + \frac 83\right)\mathcal L^3.
\end{align*}

\section{Proof of Lemma \ref{L4.2}} \label{sec:app.B}
\setcounter{equation}{0}
In this appendix, we present a proof of Lemma \ref{L4.2} in which the differential inequality for the aggregation functional $\mathcal F$ in \eqref{Z-52} will be derived. 

\begin{lemma} \label{LB.1}
Let $\{(U_i,V_i)\}$ be a solution to \eqref{Z-0}. Then, $d(U,\tilde U)$ and $d(V,\tilde V)$ satisfy 
\begin{align} \label{BB-0}
\begin{aligned}
\frac{\d}{\dt} d(U,\tilde U) &\leq -2m\kp d(U,\tilde U) + 4m\kp \mathcal L d(U,\tilde U) + 6\mathcal L(d(U,\tilde U) + \mathcal S(V,\tilde V) ) \\
&\hspace{0.5cm}+ 2\mathcal L^2 (4d(U,\tilde U) + \mathcal S(V,\tilde V)) + 4\sqrt n \mathcal S(V,\tilde V). \\
\frac{\d}{\dt} d(V,\tilde V) &\leq -2n\kp d(V,\tilde V) + 4n\kp \mathcal L d(V,\tilde V) + 6\mathcal L(d(V,\tilde V) + \mathcal S(U,\tilde U) ) \\
&\hspace{0.5cm}+ 2\mathcal L^2 (4d(V,\tilde V) + \mathcal S(U,\tilde U)) + 4\sqrt m \mathcal S(U,\tilde U).
\end{aligned}
\end{align}
\end{lemma}
\begin{proof} First, we recall \eqref{AA-4}:
\begin{align*} 
\frac{\d}{\dt} S_{ij} & = -2 m \kp S_{ij} + \frac{m\kp}{N} \sum_{k=1}^N S_{ik} S_{ij} + S_{ij}S_{kj}   + \mi (H_i - H_j) + \mi ( S_{ij} H_j - H_i S_{ij} ) \\
& \hspace{0.5cm}+ \frac\kp N \sum_{k=1}^N (c_{ik} - m)S_{kj} - (c_{ki}-m)S_{ik} - (c_{ik} - m )S_{ij} + (c_{ki}-m)S_{ik}S_{ij} \\
&\hspace{0.5cm}+\frac\kp N \sum_{k=1}^N (c_{kj} - m )S_{ik} - (c_{jk}-m)S_{kj} - (c_{jk}-m)S_{ij} + (c_{jk}-m)S_{ij}S_{kj} \\
&\hspace{0.5cm}+\frac\kp N \sum_{k=1}^N (c_{ki} - c_{ik}  +  c_{jk} - c_{kj})I_n.
\end{align*}
Thus, $S_{ij} - \tilde S_{ij}$ satisfies 
\begin{align} \label{BB-2}
\begin{aligned}
&\frac{\d}{\dt} (S_{ij}- \tilde S_{ij} ) =  -2m\kp (S_{ij}  -\tilde S_{ij}) + \frac{m\kp}{N}\sum_{k=1}^N  \underbrace{(S_{ik} S_{ij} + S_{ij} S_{kj} -\tilde S_{ik} \tilde S_{ij} - \tilde S_{ij} \tilde S_{kj})}_{=:\mathcal I_{41}} \\
&\hspace{1cm} + \underbrace{\mi ((S_{ij} - \tilde S_{ij})H_j - H_i (S_{ij} - \tilde S_{ij}))}_{=:\mathcal I_{42}}  + \frac\kp N \sum_{k=1}^N ( \mathcal I_{43} - \tilde {\mathcal I}_{43}) + \frac\kp N \sum_{k=1}^N ( \mathcal I_{44} - \tilde {\mathcal I}_{44}) \\
&\hspace{1cm} + \frac\kp N \sum_{k=1}^N \underbrace{((c_{ki} - c_{ik} + c_{jk} - c_{kj} )- (\tilde c_{ki} -\tilde  c_{ik} + \tilde c_{jk} -\tilde  c_{kj} )) I_n}_{=:\mathcal I_{45}},
\end{aligned}
\end{align} 
where $\mathcal I_{43}$ and $\mathcal I_{44}$ are defined as 
\begin{align*}
&\mathcal I_{43} := (c_{ik} - m)S_{kj} - (c_{ki}-m)S_{ik} - (c_{ik} - m )S_{ij} + (c_{ki}-m)S_{ik}S_{ij}, \\
&\mathcal I_{44} := (c_{kj} - m )S_{ik} - (c_{jk}-m)S_{kj} - (c_{jk}-m)S_{ij} + (c_{jk}-m)S_{ij}S_{kj}.
\end{align*}
Below, we provide the estimates for $\mathcal I_{4k},~k=1,\cdots,5$, respectively. \newline

\noindent $\bullet$ (Estimate of $\mathcal I_{41}$): Note that 
\begin{align*}
\mathcal I_{41} &= S_{ik} S_{ij} + S_{ij} S_{kj} -\tilde S_{ik} \tilde S_{ij} - \tilde S_{ij} \tilde S_{kj}  \\
& = S_{ik}(S_{ij} - \tilde S_{ij}) + (S_{ik} - \tilde S_{ik} ) \tilde S_{ij} + S_{ij} ( S_{kj} - \tilde S_{kj}) + (S_{ij} - \tilde S_{ij} ) \tilde S_{kj} .
\end{align*}
This yields
\begin{equation} \label{BB-5}
\|\mathcal I_{41}\|_\tF \leq 4\mathcal L d(U,\tilde U).
\end{equation}

\vspace{0.2cm}

\noindent $\bullet$ (Estimate of $\mathcal I_{42}$): Note that for a skew-Hermitian matrix $\Omega$ and a matrix $A$,  one has
\begin{align*}
\textup{tr} (\Omega A A^\dg ) = 0.
\end{align*}
This implies
\begin{align*}
\textup{Re} \textup{tr} ( \mathcal I_{42}( S_{ij} - \tilde S_{ij})^\dg ) =  \textup{Re} \textup{tr} ( \mi H_j ( S_{ij} - \tilde S_{ij}) ( S_{ij} - \tilde S_{ij})^\dg ) =0.
\end{align*}

\vspace{0.2cm}

\noindent $\bullet$ (Estimate of $\mathcal I_{43}$): We observe
\begin{align*}
\begin{aligned}
&\| (c_{ik} -m )S_{kj} - (\tilde c_{ik} - m) \tilde S_{kj} \|_\tF \\
& \hspace{1cm} = \|(c_{ik} - m) (S_{kj} - \tilde S_{kj}) + (c_{ik} - \tilde c_{ik}) \tilde S_{kj} \|_\tF \\
& \hspace{1cm} \leq \mathcal L d(U,\tilde U) + \mathcal L \mathcal S(V,\tilde V) = \mathcal L(d(U,\tilde U) + \mathcal S(V,\tilde V)).
\end{aligned}
\end{align*}
Moreover, we use \eqref{BB-5} 
\begin{align*}
& \|(c_{ki} - m )S_{ik} S_{ij} - (\tilde c_{ki} - m )\tilde S_{ik}\tilde S_{ij} \|_\tF \\
& \hspace{1cm}  = \|(c_{ki} - m) (S_{ik} S_{ij}- \tilde S_{ik} \tilde S_{ij}) + (c_{ik} - \tilde c_{ik}) \tilde S_{ik} \tilde S_{ij} \|_\tF \\
& \hspace{1cm} \leq \mathcal L \cdot 4\mathcal L d(U,\tilde U) + \mathcal S(V,\tilde V) \mathcal L^2 = \mathcal L^2 ( 4 d(U,\tilde U) + \mathcal S(V,\tilde V))
\end{align*}
to obtain
\begin{align*}
\|\mathcal I_{43} - \tilde {\mathcal I}_{43}\|_\tF \leq  3\mathcal L(d(U,\tilde U) + \mathcal S(V,\tilde V) ) + \mathcal L^2 ( 4 d(U,\tilde U) + \mathcal S(V,\tilde V)). 
\end{align*}

\vspace{0.2cm}

\noindent $\bullet$ (Estimate of $\mathcal I_{44}$): Similar to $\mathcal I_{43}$, one finds
\begin{align*}
\|\mathcal I_{44} - \tilde {\mathcal I}_{44}\|_\tF \leq  3\mathcal L(d(U,\tilde U) + \mathcal S(V,\tilde V) ) + \mathcal L^2 ( 4 d(U,\tilde U) + \mathcal S(V,\tilde V)). 
\end{align*}

\vspace{0.2cm}

\noindent $\bullet$ (Estimate of $\mathcal I_{45}$): We directly find
\begin{align*}
\|\mathcal I_{45} \|_\tF \leq 4\sqrt n \mathcal S(V,\tilde V).
\end{align*}
In \eqref{BB-2}, we combine all the estimates to obtain 
\begin{align*}
\frac{\d}{\dt} d(U,\tilde U) &\leq -2m\kp d(U,\tilde U) + 4m\kp \mathcal L d(U,\tilde U) + 6\mathcal L(d(U,\tilde U) + \mathcal S(V,\tilde V) ) \\
&\hspace{0.5cm}+ 2\mathcal L^2 (4d(U,\tilde U) + \mathcal S(V,\tilde V)) + 4\sqrt n \mathcal S(V,\tilde V).
\end{align*}
By exchanging the roles of $\mathcal U$ and $\mathcal V$, we find the differential inequality for $d(V,\tilde V)$:
\begin{align*}
\frac{\d}{\dt} d(V,\tilde V) &\leq -2n\kp d(V,\tilde V) + 4n\kp \mathcal L d(V,\tilde V) + 6\mathcal L(d(V,\tilde V) + \mathcal S(U,\tilde U) ) \\
&\hspace{0.5cm}+ 2\mathcal L^2 (4d(V,\tilde V) + \mathcal S(U,\tilde U)) + 4\sqrt m \mathcal S(U,\tilde U).
\end{align*}
\end{proof}
Below, we derive differential inequalities for $\mathcal S(U,\tilde U)$ and $\mathcal S(V,\tilde V)$. 
\begin{lemma} \label{LB.2} 
Let $\{(U_i,V_i)\}$ be  a solution to \eqref{Z-0}. Then, $\mathcal S(U,\tilde U)$ and $\mathcal S(V,\tilde V)$ satisfy
\begin{align} \label{BB-14} 
\begin{aligned}
\frac{\d}{\dt} \mathcal S(U,\tilde U) &\leq -2m\kp \mathcal S(U,\tilde U) + 4\mathcal L d(U,\tilde U) + \mathcal D( \mathcal H)d (U,\tilde U) + 6\mathcal L(\mathcal S(U,\tilde U) + \mathcal S(V,\tilde V)) \\
&\hspace{0.5cm} + 8\mathcal L^2 d(U,\tilde U) + 2\mathcal L^2 \mathcal S(V,\tilde V) + 4\sqrt n \mathcal S(V,\tilde V). \\
\frac{\d}{\dt} \mathcal S(V,\tilde V) &\leq -2n\kp \mathcal S(V,\tilde V) + 4\mathcal L d(V,\tilde V) + \mathcal D( \mathcal H)d (V,\tilde V) + 6\mathcal L(\mathcal S(V,\tilde V) + \mathcal S(U,\tilde U)) \\
&\hspace{0.5cm} + 8\mathcal L^2 d(V,\tilde V) + 2\mathcal L^2 \mathcal S(U,\tilde U) + 4\sqrt m \mathcal S(U,\tilde U).
\end{aligned}
\end{align}
\end{lemma}
\begin{proof}
We recall \eqref{AA-15}: 
\begin{align*} \label{BB-15}
\begin{aligned}
\frac{\d}{\dt} (n- d_{ij}) &= -2m\kp ( n-d_{ij}) + \frac{m\kp}{N}\sum_{k=1}^N \textup{tr} (S_{ik}S_{ij} + S_{ij} S_{kj}) \\
&\hspace{0.5cm} + \mi \textup{tr}(H_i - H_j) + \mi \textup{tr}(S_{ij} H_j - H_i S_{ij}) \\
&\hspace{0.5cm}+ \frac\kp N \sum_{k=1}^N (c_{ik} - m)(n-d_{kj}) - (c_{ki}-m)(n-d_{ik}) - (c_{ik}-m)(n-d_{ij}) \\
&\hspace{0.5cm}+\frac\kp N \sum_{k=1}^N (c_{kj}-m)(n-d_{ik}) - (c_{jk}-m)(n-d_{kj}) - (c_{jk}-m)(n-d_{ij}) \\
&\hspace{0.5cm} +\frac\kp N \sum_{k=1}^N (c_{ki}-m)\textup{tr} (S_{ik}S_{ij}) + (c_{jk}-m)\textup{tr}( S_{ij}S_{kj}) \\
&\hspace{0.5cm} + \frac\kp N \sum_{k=1}^N[  (c_{ki} - c_{ik}) + (c_{jk} - c_{kj})]\sqrt n. 
\end{aligned}
\end{align*}
We denote
\begin{equation*}
p_{ij}:= n-d_{ij},\quad q_{ij} := m-c_{ij}.
\end{equation*}
Then, $p_{ij}- \tilde p_{ij}$ satisfies
\begin{align}  \label{BB-17}
\begin{aligned}
&\frac{\d}{\dt} (p_{ij} - \tilde p_{ij}) = -2m\kp (p_{ij} - \tilde p_{ij})+ \frac{m\kp}{N} \sum_{k=1}^N \underbrace{\textup{tr}(S_{ik} S_{ij} + S_{ij}S_{kj} -\tilde S_{ik} \tilde S_{ij} - \tilde S_{ij} \tilde S_{kj})}_{=:\mathcal I_{51}} \\
&\hspace{0.2cm} +\underbrace{ \mi \textup{tr} ((S_{ij} - \tilde S_{ij})(H_j - H_i) ) }_{=:\mathcal I_{52}} + \frac\kp N \sum_{k=1}^N ( \mathcal I_{53} - \tilde {\mathcal I}_{53} ) + \frac\kp N \sum_{k=1}^N ( \mathcal I_{54} - \tilde {\mathcal I}_{54} ) \\
& \hspace{0.2cm} + \frac\kp N \sum_{k=1}^N ( \mathcal I_{55} - \tilde {\mathcal I}_{55} )  + \frac\kp N \sum_{k=1}^N \underbrace{( p_{ik} - \tilde p_{ik} + p_{kj} - \tilde p_{kj}) + (\tilde p_{ki} - \tilde p_{ki}) + ( \tilde p_{jk} - p_{jk})}_{=:\mathcal I_{56}} \sqrt n . 
\end{aligned}
\end{align}
Below, we present the estimates for $\mathcal I_{5k},~k=1,\cdots,6$, separately. \newline

\noindent $\bullet$ (Estimate of $\mathcal I_{51}$): We recall \eqref{BB-5} to find
\begin{equation*}
|\mathcal I_{51}| \leq 4\mathcal L d(U,\tilde U). 
\end{equation*}
$\bullet$ (Estimate of $\mathcal I_{52}$): we set
\begin{equation*}
\mathcal D(\mathcal H) := \max_{1\leq i,j\leq N} \|H_i - H_j\|_\infty,
\end{equation*}
Then, one has
\begin{equation*}
|\mathcal I_{52}| \leq  \mathcal D(\mathcal H) d(U,\tilde U).
\end{equation*}
$\bullet$ (Estimate of $\mathcal I_{53}$): Note that 
\begin{align*}
&|(c_{ik} - m)(n-d_{kj}) - (\tilde c_{ik} -m)(n- \tilde d_{kj})| \\
& \hspace{1cm} \leq |p_{ik} - \tilde p_{ik}| |q_{ik}| + |\tilde p_{ik}| |q_{kj} - \tilde q_{kj}| \leq \mathcal L ( \mathcal S(U,\tilde U) + \mathcal S(V,\tilde V)).
\end{align*}
Thus, we find
\begin{align*}
\|\mathcal I_{53} - \tilde {\mathcal I}_{53} \|_\tF \leq 3 \mathcal L ( \mathcal S(U,\tilde U) + \mathcal S(V,\tilde V)). 
\end{align*}
$\bullet$ (Estimate of $\mathcal I_{54}$): We observe
\begin{align*}
\|\mathcal I_{54} - \tilde {\mathcal I}_{54} \|_\tF \leq  3 \mathcal L ( \mathcal S(U,\tilde U) + \mathcal S(V,\tilde V)). 
\end{align*}
$\bullet$ (Estimate of $\mathcal I_{55}$): Note that 
\begin{align*}
& q_{ki} \textup{tr} (S_{ik} S_{ij} - \tilde q_{ki} \textup{tr} (\tilde S_{ik} \tilde S_{ij}) \\
& \hspace{0.5cm} = q_{ij} \textup{tr} (S_{ik} S_{ij} - \tilde S_{ik} \tilde S_{ij}) + (q_{ij} - \tilde q_{ij}) \textup{tr}(\tilde S_{ik} \tilde S_{ij}) \leq 4\mathcal L^2 d(U,\tilde U) + \mathcal L^2 \mathcal S(V,\tilde V) . 
\end{align*}
Thus, we have
\begin{align*}
\|\mathcal I_{55} - \tilde {\mathcal I}_{55} \|_\tF \leq 8\mathcal L^2 d(U,\tilde U) +2 \mathcal L^2 \mathcal S(V,\tilde V). 
\end{align*}
$\bullet$ (Estimate of $\mathcal I_{56}$): We directly find 
\begin{align*}
|\mathcal I_{56} | \leq 4\sqrt n \mathcal S(V,\tilde V).
\end{align*}
In \eqref{BB-17}, we collect all the estimates to find
\begin{align*}
\frac{\d}{\dt} \mathcal S(U,\tilde U) &\leq -2m\kp \mathcal S(U,\tilde U) + 4\mathcal L d(U,\tilde U) + \mathcal D( \mathcal H)d (U,\tilde U) + 6\mathcal L(\mathcal S(U,\tilde U) + \mathcal S(V,\tilde V)) \\
&\hspace{0.5cm} + 8\mathcal L^2 d(U,\tilde U) + 2\mathcal L^2 \mathcal S(V,\tilde V) + 4\sqrt n \mathcal S(V,\tilde V).
\end{align*}
By exchanging the roles of $\mathcal U$ and $\mathcal V$, we can derive a differential inequality for $\mathcal S(V,\tilde V)$: 
\begin{align*}
\frac{\d}{\dt} \mathcal S(V,\tilde V) &\leq -2n\kp \mathcal S(V,\tilde V) + 4\mathcal L d(V,\tilde V) + \mathcal D( \mathcal H)d (V,\tilde V) + 6\mathcal L(\mathcal S(V,\tilde V) + \mathcal S(U,\tilde U)) \\
&\hspace{0.5cm} + 8\mathcal L^2 d(V,\tilde V) + 2\mathcal L^2 \mathcal S(U,\tilde U) + 4\sqrt m \mathcal S(U,\tilde U).
\end{align*}
\end{proof}

\vspace{0.5cm}

Now, we are ready to provide a proof of Lemma \ref{L4.2} with the help of Lemma \ref{LB.1} and Lemma \ref{LB.2}. \newline

\noindent {\it  Proof of Lemma \ref{L4.2}:} ~recall the inequalities in \eqref{BB-0} and \eqref{BB-14}:
\begin{align} \label{BB-25}
\begin{aligned}
\frac{\d}{\dt} d(U,\tilde U) &\leq -2m\kp d(U,\tilde U) + 4m\kp \mathcal L d(U,\tilde U) + 6\mathcal L(d(U,\tilde U) + \mathcal S(V,\tilde V) ) \\
&\hspace{0.5cm}+ 2\mathcal L^2 (4d(U,\tilde U) + \mathcal S(V,\tilde V)) + 4\sqrt n \mathcal S(V,\tilde V), \\
\frac{\d}{\dt} d(V,\tilde V) &\leq -2n\kp d(V,\tilde V) + 4n\kp \mathcal L d(V,\tilde V) + 6\mathcal L(d(V,\tilde V) + \mathcal S(U,\tilde U) ) \\
&\hspace{0.5cm}+ 2\mathcal L^2 (4d(V,\tilde V) + \mathcal S(U,\tilde U)) + 4\sqrt m \mathcal S(U,\tilde U), \\
\frac{\d}{\dt} \mathcal S(U,\tilde U) &\leq -2m\kp \mathcal S(U,\tilde U) + 4\mathcal L d(U,\tilde U) + \mathcal D( \mathcal H)d (U,\tilde U) + 6\mathcal L(\mathcal S(U,\tilde U) + \mathcal S(V,\tilde V)) \\
&\hspace{0.5cm} + 8\mathcal L^2 d(U,\tilde U) + 2\mathcal L^2 \mathcal S(V,\tilde V) + 4\sqrt n \mathcal S(V,\tilde V), \\
\frac{\d}{\dt} \mathcal S(V,\tilde V) &\leq -2n\kp \mathcal S(V,\tilde V) + 4\mathcal L d(V,\tilde V) + \mathcal D( \mathcal G)d (V,\tilde V) + 6\mathcal L(\mathcal S(V,\tilde V) + \mathcal S(U,\tilde U)) \\
&\hspace{0.5cm} + 8\mathcal L^2 d(V,\tilde V) + 2\mathcal L^2 \mathcal S(U,\tilde U) + 4\sqrt m \mathcal S(U,\tilde U).
\end{aligned}
\end{align}
We add all the inequalities in \eqref{BB-25} to obtain the desired inequality for $\mathcal F$:
\begin{align*}
\frac{\d \mathcal F}{dt} &\leq -\kp \left( 2   m - 8  \sqrt n - \frac{\mathcal D(\mathcal H)}{\kp} \right) \mathcal F + 4(n+1)\kp \mathcal L (d(U,\tilde U) + d(V,\tilde V))  + 6\kp \mathcal L \mathcal F  \\
& \hspace{0.5cm}+ 12\kp \mathcal L ( \mathcal S(U,\tilde U) + \mathcal S(V,\tilde V)) + 16\kp \mathcal L^2( d(U,\tilde U) +  d(V,\tilde V) ) + 4\kp \mathcal L^2( \mathcal S(U,\tilde U) + \mathcal S(V,\tilde V))  \\
&\leq -\kp \left( 2   m - 8  \sqrt n - \frac{ \max\{ \mathcal D(\mathcal H), \mathcal D(\mathcal G)\} }{\kp} \right) \mathcal F   + \kp ( 4n + 22) \mathcal L \mathcal F + 20\kp \mathcal L^2 \mathcal F.
\end{align*}

\section{Emergent dynamics of the DSOM model} \label{sec:app.C}
\setcounter{equation}{0}
In this section, we replace the product of two unitary groups $\Un \times \Um$ by the product of two special orthogonal groups $\SOn \times \SOm$. Since all elements of special orthogonal matrices are real-valued, one has for any (real-valued) square matrices $A$ and $B$, 
\begin{equation} \label{Y-0} 
\langle A,B\rangle_\tF = \tr(AB^\top) = \tr (BA^\top) = \langle B,A\rangle_\tF.
\end{equation}
Thus, model \eqref{C-10-1} reduces to 
\begin{equation} \label {Y-1}
\begin{cases}
\displaystyle \dot U_j = B_j U_j + \frac\kp N \sum_{k=1}^N \langle V_j,V_k\rangle_\tF ( U_k - U_jU_k^\top  U_j), \quad t > 0, \\
\displaystyle \dot V_j = C_j V_j + \frac\kp N \sum_{k=1}^N \langle U_j,U_k\rangle_\tF (V_k - V_j V_k^\top V_j), \\
\displaystyle (U_j, V_j)(0) = (U_j^0,V_j^0) \in \SOn \times \SOm,\quad 1\leq j \leq N,
\end{cases}
\end{equation}
where $B_j \in \bbr^{n\times n \times n \times n } $ and $C_j \in \bbr^{m \times m \times m \times m }$ are given rank-4 tensors satisfying skew-symmetric properties:
\begin{equation*}
[B_j]_{\alpha_1\beta_1\alpha_2\beta_2} = -[ B_j]_{\alpha_2\beta_2\alpha_1\beta_1}, \quad [C_j]_{\gamma_1\delta_1\gamma_2\delta_2}=-[{C}_j]_{\gamma_2\delta_2\gamma_1\delta_1}. 
\end{equation*}

\begin{remark}
Due to the symmetric property of the Frobenius inner product on  real-valued matrices, one has
\begin{equation} \label{Y-2}
d_{ij} = \langle U_i,U_j\rangle_\tF  = \langle U_j,U_i\rangle_\tF=d_{ji},\quad c_{ij}= \langle V_i,V_j\rangle_\tF = \langle V_j,V_i\rangle_\tF =c_{ji}. 
\end{equation}
Due to  the relation \eqref{Y-2}, for example, the last term $c_{ki} - c_{ik} + c_{jk} - c_{kj}$ in \eqref{AA-3} vanishes. Thus, it would be expected that the relation \eqref{Y-2} relaxes the condition such as the dimension condition $n\geq m > 4\sqrt n$ in \eqref{Z-10}\textup{(i)} (see also Remark \ref{RA.1}).  
\end{remark}

The contents of this appendix are exactly the same as those of Section \ref{sec:4}, and we only provide a sketch of the proof. In Section \ref{sec:5.1},  we are concerned with the identical (or homogeneous) Hamiltonian, and in Section \ref{sec:5.2}, the non-identical (or heterogeneous) Hamiltonian is considered to exhibit the phase-locked state.  

\subsection{Homogeneous ensemble} \label{sec:5.1}
In this subsection, thanks to the solution splitting property, we assume that
\begin{equation} \label{Y-3}
B_j \equiv O,\quad C_j \equiv O,\quad j=1,\cdots,N. 
\end{equation}
Below, we state the main theorem for an identical ensemble under which the complete aggregation occurs. For this, we recall diameters and aggregation functional:
\begin{align*}
\begin{aligned}
&\mathcal D(\mathcal U(t))  = \max_{1\leq i,j\leq N} \|U_i(t) - U_j(t)\|_\tF,\quad  \mathcal  S(\mathcal U(t)) = \max_{1\leq i,j\leq N} |n- d_{ij}(t)|, \\
&\mathcal D(\mathcal V(t))  = \max_{1\leq i,j\leq N} \|V_i(t) - V_j(t)\|_\tF, \quad \mathcal S(\mathcal V(t)) = \max_{1\leq i,j\leq N} |m- c_{ij}(t)|, \\
&\mathcal L = \mathcal D(\mathcal U) +  \mathcal D(\mathcal V) +  \mathcal S(\mathcal U) +  \mathcal S(\mathcal V).
\end{aligned}
\end{align*}
\begin{theorem} \label{TC.1} 
Suppose initial data satisfy 
\begin{equation} \label{Y-5} 
\mathcal L^0 < \alpha_{m,n}  =\frac{ -(12n+27) +\sqrt{(12n+27)^2 +24m(3n+4)}}{4(4n+3)},
\end{equation}
and let $\{(U_i, V_i)\}$ be a solution to \eqref{Y-1} with \eqref{Y-3}. Then, we have
\begin{equation*}
\lim_{t\to\infty} \mathcal L(t) =0.
\end{equation*}
\end{theorem}

\begin{proof}
By recalling the inequalities in \eqref{AA-2} with the relation \eqref{Y-0}, we see
\begin{align} \label{Y-6}
\begin{aligned}
&\frac{\d}{\dt} \mathcal D(\mathcal U) \leq -2m \kp \mathcal D(\mathcal U) + m\kp \mathcal D(\mathcal U)^3 + 6\kp \mathcal S(\mathcal V) \mathcal D(\mathcal U) + 2\kp \mathcal S(\mathcal V) \mathcal D(\mathcal U)^2  , \\
&\frac{\d}{\dt} \mathcal D(\mathcal V) \leq -2n \kp \mathcal D(\mathcal V) + n\kp \mathcal D(\mathcal V)^3 + 6\kp \mathcal S(\mathcal U) \mathcal D(\mathcal V) + 2\kp \mathcal S(\mathcal U) \mathcal D(\mathcal V)^2 .
\end{aligned}
\end{align} 
Similarly, we use \eqref{AA-12} to find
\begin{align} \label{Y-7}
\begin{aligned}
&\frac{\d}{\dt} \mathcal S(\mathcal U) \leq -2m\kp \mathcal S(\mathcal U) + 2m\kp \mathcal D( \mathcal U)^2 + 6\kp \mathcal S(\mathcal U) \mathcal S(\mathcal V) + 2\kp \mathcal S(\mathcal V) \mathcal D(\mathcal U)^2 , \\
&\frac{\d}{\dt} \mathcal S(\mathcal V) \leq -2n\kp \mathcal S(\mathcal V) + 2n\kp \mathcal D( \mathcal V)^2 + 6\kp \mathcal S(\mathcal U) \mathcal S(\mathcal V) + 2\kp \mathcal S(\mathcal U) \mathcal D(\mathcal V)^2 .
\end{aligned}
\end{align}
Now, we add \eqref{Y-6} and \eqref{Y-7} to obtain a differential inequality for $\mathcal L$: 
\begin{equation} \label{Y-8}
\frac\d\dt \mathcal L  \leq -2\kp m \mathcal L + \kp (4n+9)\mathcal L^2 + \kp \left( 2n +\frac83 \right) \mathcal L^3. 
\end{equation}
If we introduce an auxiliary polynomial:
\begin{equation} \label{Y-9}
f(s) := \left( 2n+\frac83\right) s^2 + (4n+9)s - 2m,
\end{equation}
then \eqref{Y-8} can be written as
\begin{equation*}
\frac\d\dt \mathcal L \leq \kp \mathcal Lf(\mathcal L).
\end{equation*}
Furthermore, we notice that $\alpha_{m,n}$ in \eqref{Y-5} becomes a unique positive root of the quadratic polynomial $f$ in \eqref{Y-9}.  Hence, the desired zero convergence of $\mathcal L$ directly follows from the dynamical systems theory. 
\end{proof}

\subsection{Heterogeneous ensemble} \label{sec:5.2}
In this subsection, we are concerned with the heterogeneous Hamiltonians where the natural frequency tensors are different in general. Similar to Section \ref{sec:4.2}, there exist two skew-symmetric matrices $\Omega_j \in \bbr^{n\times n}$ and $\Psi_j \in \bbr^{m\times m}$ such that
\begin{equation*}
B_j U_j = \Omega_j U_j,\quad C_j V_j = \Psi_j V_j,\quad 1\leq j \leq N.
\end{equation*}
Thus, our model reads as 
\begin{equation} \label {Y-10}
\begin{cases}
\displaystyle \dot U_j = \Omega_j U_j + \frac\kp N \sum_{k=1}^N \langle V_j,V_k\rangle_\tF ( U_k - U_jU_k^\top  U_j), \\
\displaystyle \dot V_j = \Psi_j V_j + \frac\kp N \sum_{k=1}^N \langle U_j,U_k\rangle_\tF (V_k - V_j V_k^\top V_j), \\
\displaystyle (U_j, V_j)(0) = (U_j^0,V_j^0) \in \SOn \times \SOm,\quad 1\leq j \leq N,
\end{cases}
\end{equation}
A crucial estimate used for the emergence of the phase-locked state is \eqref{Z-55} stated in Lemma \ref{L4.2} where the aggregation functional $\mathcal F$ tends to zero. Here, $\mathcal F$ measures the inter-distance between any two solution configurations $\{(U_i,V_i)\}$ and $\{ (\tilde U_i, \tilde V_i)\}$:
\begin{equation*}
\mathcal F(t) = d(U,\tilde U)(t) + d(V,\tilde V)(t) + \mathcal S(U,\tilde U)(t) + \mathcal S(V,\tilde V)(t),
\end{equation*}
where the diameters are defined in \eqref{Z-50-2} and \eqref{Z-50-6}. We also denote the diameters for natural frequencies:
\begin{equation*}
\mathcal D(\Omega) :=\max_{1\leq i,j\leq N}\| \Omega_i - \Omega_j\|_\infty, \quad \mathcal D(\Psi) :=\max_{1\leq i,j\leq N}\| \Psi_i - \Psi_j\|_\infty.
\end{equation*}
Without loss of generality, we may assume
\begin{equation*}
\mathcal D(\Omega)  \geq \mathcal D(\Psi).
\end{equation*}
In the following lemma, we provide a temporal evolution of $\mathcal F$. 

\begin{lemma} \label{LC.1} 
Let  $\{(U_i,V_i)\}$ and $\{ (\tilde U_i, \tilde V_i)\}$ be any two solutions to \eqref{Y-10}. Then, the aggregation functionals  $\mathcal L$ and $\mathcal F$ satisfy 
\begin{align} \label{Y-15} 
\begin{aligned}
&\frac{\d \mathcal L}{dt} \leq   2(1+3\sqrt n ) \mathcal D(\Omega) -2\kp m \mathcal L +\kp(4n+9)\mathcal L^2 + \kp \left( 2n + \frac83\right) \mathcal L^3, \quad t > 0, \\
&\frac{\d \mathcal F}{dt} \leq  -\kp \left( 2m - \frac{\max \{ \mathcal D(\Omega),\mathcal D(\Psi)\}}{\kp} \right) \mathcal F + \kp(4n+22)\mathcal L \mathcal F + 20\kp \mathcal L^2\mathcal F.
\end{aligned}
\end{align}
\end{lemma}
\begin{proof}
For the inequality for $\mathcal L$, it directly follows from \eqref{Y-8} in Theorem \ref{TC.1}. Similarly, if we closely follow a proof of Lemma \ref{L4.2} presented in Appendix \ref{sec:app.B}, then we find the desired inequality for $\mathcal F$. 
\end{proof} 
By applying $\eqref{Y-15}_2$, we establish the desired  practical aggregation estimate for \eqref{Y-10}. 
\begin{proposition} \label{PC.1} 
Suppose that the system parameters and initial data satisfy 
\begin{equation} \label{Y-18} 
\mathcal D(\Omega) \geq \mathcal D(\Psi),\quad \kp > \kp_\textup{c},\quad \mathcal L^0 < \nu_2,
\end{equation}
where $\kp_\textup{c}$ and $\nu_2$ are specified in \eqref{Y-24} and \eqref{Y-25}, respectively, and let $\{(U_i,V_i)\}$ be a solution to \eqref{Y-10}. Then, one has
\begin{equation*}
\lim_{\kp \to\infty} \limsup_{t\to\infty} \mathcal L(t) = 0.
\end{equation*}

\end{proposition}

\begin{proof}
We introduce an auxiliary polynomial:
\begin{equation} \label{Y-19}
g(s):= 2 m s - (4n+9)s^2 - \left( 2n+\frac83\right) s^3.
\end{equation}
Then, we notice that $g$ has three roots, say, $\alpha_0 <0 <\alpha_1$. Since $\eqref{Y-15}_2$ is rewritten as
\begin{equation*}
\dot{ \mathcal L} \leq \kp \left(   \frac{2(1+3\sqrt n) \mathcal D(\mathcal H)}{\kp} - g(\mathcal L)\right) =:\kp p(s),
\end{equation*}
for a sufficiently large $\kp>0$, $p$ admits one negative root, say, $\nu_0<0$ and two  positive roots, say, $0<\nu_1<\nu_2$ with continuous dependence on $\kp$: for $\nu_1=\nu_1(\kp)$ and $\nu_2=\nu_2(\kp)$, 
\begin{equation*}
\lim_{\kp\to\infty} \nu_1(\kp) = 0,\quad  \lim_{\kp \to\infty} \nu_2(\kp) = \alpha_1. 
\end{equation*}
Since we assume $\eqref{Y-18}_3$, it follows from dynamical systems theory that  there exists a finite entrance time $T_*>0$ such that
\begin{equation*}
\mathcal L(t) < \nu_1,\quad t>T_*.
\end{equation*}
Hence, this yields the desired result. 
\end{proof} 
\begin{remark}
For the explicit value of $\kp_\textup{c}$, we see that the cubic polynomial $g$ in \eqref{Y-19} admits the local maximum at $s=s_*$:
\begin{equation*}
s_* = \frac{ -(4n+9) + \sqrt{ (4n+9)^2 + 4m(3n+4)}}{6n+8}.
\end{equation*} 
Hence, $\kp_\textup{c}$ is chosen to be
\begin{equation} \label{Y-24}
\kp_\textup{c} := \frac{ 2(1+3\sqrt n)\mathcal D(\Omega)}{g(s_*)} \quad \textup{so that}\quad g(s_*) > \frac{2(1+3\sqrt n)}{\kp_\textup{c}}.
\end{equation}
For this $\kp_\textup{c}$, $\nu_2=\nu_2(\kp)$ is completely determined as the largest positive root of 
\begin{equation} \label{Y-25}
p(s) =0 \quad \textup{for $\kp>\kp_\textup{c}$}. 
\end{equation}
\end{remark}

\vspace{0.5cm}

It follows from Proposition \ref{P5.1} that we can make $\mathcal L$ sufficiently small by increasing the value of $\kp$ for $t>T_*$. Thus, there exists $\kp_\textup{p} >\kp_\textup{c}$ such that for $\kp>\kp_\textup{p}$, 
\begin{equation*}
(4n+22)\mathcal L + 20\mathcal L^2 < \frac12 \left( 2m -\frac{\max\{ \mathcal D(\Omega),\mathcal D(\Psi)\}}{\kp} \right) =: \frac12\Lambda. 
\end{equation*}
Hence, $\eqref{Y-15}_2$ yields
\begin{equation*}
\frac{d \mathcal F}{dt} \leq -\frac12\Lambda \mathcal F,\quad t>T_*. 
\end{equation*}
Since we have established the zero convergence of $\mathcal F$, we can obtain the same result in Theorem \ref{T4.2}. Thanks to exactly the same proof, we only state the results. 

\begin{theorem} \label{TC.2}
Suppose that the system parameters and initial data satisfy 
\begin{equation*}
n\geq m,\quad \mathcal D(\Omega) \geq \mathcal D(\Psi),\quad \kp>\kp_\textup{p},\quad \max\{\mathcal L^0,\tilde{\mathcal L}^0\} \leq \nu_2,
\end{equation*}
and let $\{(U_i,V_i)\}$ and $\{(\tilde U_i,\tilde V_i)\}$ be any two solutions to \eqref{Y-10}. Then, the following assertions hold. \newline
\begin{enumerate}
\item
The aggregation functional $\mathcal F$ converges to zero exponentially. 
\item
The normalized velocities $\dot U_i U_i^\top$ and $\tilde{\dot U}_i \tilde U_i^\top$ synchronize:
\begin{equation*}
\| \dot U_i U_i^\top - \tilde{\dot U}_i \tilde U_i^\top \|_\tF \leq 2\kp ( m+n) \mathcal F. 
\end{equation*}
\item
There exist special orthogonal matrices $X_\infty \in \SOn$ and $Y_\infty \in \SOm$  such that
\begin{align*}
& \lim_{t\to\infty} U_i^\top(t) \tilde U_i(t) =X_\infty,\quad \lim_{t\to\infty} \| \tilde U_i (t)- U_i(t) X_\infty \|_\tF = 0, \\
&\lim_{t\to\infty} V_i^\top(t) \tilde V_i(t) =Y_\infty,\quad \lim_{t\to\infty} \| \tilde V_i (t)- V_i(t) Y_\infty \|_\tF = 0. 
\end{align*}
\item
System \eqref{Y-10} exhibits asymptotic phase-locking: for any indices $i$ and $j$,
\begin{equation*}
\exists~~\lim_{t\to\infty} U_i(t) U_j^\top (t) \quad \textup{and}\quad  \exists~~\lim_{t\to\infty} V_i(t) V_j^\top (t).
\end{equation*}
\end{enumerate}
\end{theorem}

\end{document}